\documentclass[11pt]{article}
\usepackage{amsmath,amssymb,amsthm} 
\usepackage{bbm}
\usepackage{stmaryrd}
\usepackage{paralist}
\usepackage{hyperref}
\usepackage{graphicx}
\usepackage{color}

\setdefaultenum{(i)}{}{}{}
\usepackage[margin=1in]{geometry} 

\numberwithin{equation}{section}

\numberwithin{equation}{section}

\renewenvironment{thebibliography}[1]{%
\begin{oldthebibliography}{#1}%
\setlength{\baselineskip}{.8em}
\linespread{.8}
\small
\setlength{\parskip}{0ex}%
\setlength{\itemsep}{.1em}%
}%
{%
\end{oldthebibliography}%
}

\theoremstyle{plain}
\newtheorem{theorem}{Theorem}[section]
\newtheorem{proposition}[theorem]{Proposition}
\newtheorem{lemma}[theorem]{Lemma}
\newtheorem{corollary}[theorem]{Corollary}
\newtheorem{definition}[theorem]{Definition}

\DeclareMathAlphabet\scr{U}{scr}{m}{n}
\SetMathAlphabet\scr{bold}{U}{scr}{b}{n}
  \DeclareFontFamily{U}{scr}{\skewchar\font'177}%
  \DeclareFontShape{U}{scr}{m}{n}{<-6>rsfs5<6-8>rsfs7<8->rsfs10}{}%
  \DeclareFontShape{U}{scr}{b}{n}{<-6>rsfs5<6-8>rsfs7<8->rsfs10}{}%

\theoremstyle{definition}
\newtheorem{remark}[theorem]{Remark}

\begin{document}

\title{Equilibrium Returns with Transaction Costs\footnote{We are grateful to Michalis Anthropelos, Peter Bank, Paolo Guasoni, and Felix K\"ubler for stimulating discussions and detailed comments. Moreover, we thank an anonymous referee for his or her careful reading and pertinent remarks.}}

\author{Bruno Bouchard\thanks{Universit\'e Paris-Dauphine, PSL, CNRS, UMR [7534], CEREMADE, 75016 Paris, France, email \texttt{bouchard@ceremade.dauphine.fr}.}
\and
Masaaki Fukasawa\thanks{Osaka University, Graduate School of Engineering Science, 1-3 Machikayama, Toyonaka, Osaka, Japan, 
email: \texttt{fukasawa@sigmath.es.osaka-u.ac.jp} and Tokyo Metropolitan University, Graduate School of Social Sciences. Support from KAKENHI Grant number 25245046 is gratefully acknowledged.} 
\and
Martin Herdegen\thanks{University of Warwick, Department of Statistics, Coventry, CV4 7AL, UK, email
 \texttt{M.Herdegen@warwick.ac.uk}. Partly supported by the Swiss National Science Foundation (SNF) under grant 150101.} 
\and
Johannes Muhle-Karbe\thanks{Carnegie Mellon University, Department of Mathematical Sciences, 5000 Forbes Avenue, Pittsburgh, PA 15213, USA, email \texttt{johannesmk@cmu.edu}. Parts of this paper were written while this author was visiting ETH Z\"urich; he is grateful to the Forschungsinstitut f\"ur Mathematik and H.M.~Soner for their hospitality.}
}

\date{\today}

\maketitle

\begin{abstract}
We study how trading costs are reflected in equilibrium returns. To this end, we develop a tractable continuous-time risk-sharing model, where heterogeneous mean-variance investors trade subject to a quadratic transaction cost. The corresponding equilibrium is characterized as the unique solution of a system of coupled but linear forward-backward stochastic differential equations. Explicit solutions are obtained in a number of concrete settings. The sluggishness of the frictional portfolios makes the corresponding equilibrium returns mean-reverting. Compared to the frictionless case, expected returns are higher if the more risk-averse agents are net sellers or if the asset supply expands over time. 
\end{abstract}




\section{Introduction}

It is empirically well documented that asset returns depend on liquidity~\cite{amihud.mendelson.86,brennan.96,pastor.03}. To understand the theoretical underpinning of such ``liquidity premia'', we study a continuous-time risk-sharing equilibrium with transaction costs.\footnote{Liquidity premia with exogenous asset prices are studied by \cite{constantinides.86,jang.al.07,lynch.tan.11,gerhold.al.14,dai.al.16}, for example.} For tractability, we assume (local) mean-variance preferences and quadratic trading costs, levied on the agents' trading rates. Then, both the \emph{unique} equilibrium return that clears the market and the corresponding optimal trading strategies can be characterized by a system of coupled but linear forward-backward stochastic differential equations (FBSDEs). These equations can be solved explicitly in terms of matrix power series, leading to closed-form expressions for the liquidity premia compared to the frictionless benchmark. 

If the risk aversions of all agents are homogenous and the asset supply remains constant over time, then the frictionless price dynamics still clear the market. As a consequence, illiquidity only affects trading strategies but not equilibrium prices in this case. By contrast, if the asset supply expands over time, positive liquidity premia are necessary to compensate the agents for the trading costs incurred when purchasing these additional shares. 

Nontrivial liquidity premia also arise with heterogenous preferences. Then, the more risk averse agents have a stronger motive to trade and therefore have to provide additional compensation to the less risk-averse ones. This leads to positive liquidity premia when the more risk averse agents are net sellers. With heterogenous preferences, illiquidity also makes expected returns mean reverting. This result does not depend on mean-reverting fundamentals, but is instead induced by the sluggishness of the frictional portfolios. With trading costs, allocations do not move directly to their stationary allocation but only gradually adjust over time, leading to autocorrelated return dynamics. For example, if endowment exposures have independent increments, then the liquidity premia have Ornstein-Uhlenbeck dynamics. If the agents exogenous trading needs are also mean-reverting, they enter the liquidity premium as a stochastic mean-reversion level. Illiquidity in turn determines the fluctuations of the actual equilibrium return around this value.

From a mathematical perspective, our analysis is based on the study of systems of coupled but linear FBSDEs. Since their forward components are degenerate, general FBSDE theory as in \cite{delarue.02} only yields local existence in this context. As we need global existence and uniqueness results, we provide a direct argument. Using the theory of primary matrix functions, we extend the univariate results of Bank, Soner, and Vo\ss~\cite{bank.al.17} to the multivariate settings needed to analyze the interaction of multiple agents trading several assets. In order to cover tractable stationary models as a special case, we also show how to extend this analysis to infinite time horizons under suitable transversality conditions.

\paragraph{Related Literature} 

Equilibrium models with transaction costs are notoriously intractable, because trading costs severely complicate the agents' individual optimization problems. Moreover, representative agents cannot be used to simplify the analysis since they abstract from the trades between the individual market participants.

Accordingly, most of the literature on equilibrium asset pricing with transaction costs has focused either on numerical methods or on models with very particular simplifying assumptions. For example, \cite{heaton.lucas.96,buss.al.13,buss.dumas.15} propose algorithms for the numerical approximation of equilibrium dynamics in discrete-time, finite-state models. 

In contrast, \cite{lo.al.04,vayanos.vila.99,weston.17} obtain explicit formulas in continuous-time models but focus on settings with deterministic asset prices for tractability. Garleanu and Pedersen~\cite{garleanu.pedersen.16} solve for the equilibrium returns in a model with a single rational agent and noise traders. For exogenous mean-reverting demands, they also obtain mean-reverting returns like in our model.\footnote{Mean-reverting fundamentals also drive the mean-reverting dynamics in the overlapping-generations model with linear costs studied in \cite{vayanos.98}, for example.} Our more general results show that this effect persists even in the absence of mean-reverting fundamentals, as the sluggishness of optimal portfolios with transaction costs already suffices to generate this effect.

A similar observation is made by Sannikov and Skrzypacz~\cite{sannikov.skrzypacz.16}. Like us, they study a model with several rational mean-variance investors. However, by making information about trading targets private, they also strive to endogenize the price impact. If trades are implemented by means of a ``conditional double auction'', where each agent observes all others' supply and demand schedules, then linear, stationary equilibria can be characterized by a coupled system of algebraic equations. However, this system generally admits multiple solutions and these are not available in closed form except in the case of (almost) homogenous risk aversions. 

\paragraph{Outline of the paper} 

This article is organized as follows.  Section~\ref{s:model} describes the model, both in its frictionless baseline version and with quadratic trading costs. In Section 3, we derive the frictionless equilibrium, before turning to individual optimality with transaction costs (and given exogenous returns) in Section~\ref{s:indopt}. Section~\ref{s:eqopt} in turn contains our main results, on existence, uniqueness, and an explicit characterization of the equilibrium return, complemented by several examples. Section \ref{sec:conclusion} concludes. Appendix~\ref{a:FBSDE} contains the existence and uniqueness results for linear FBSDEs that are used in Section ~\ref{s:indopt} and \ref{s:eqopt}. Appendix B summarizes some material on primary matrix functions that is needed in Appendix~\ref{a:FBSDE}.

\paragraph{Notation}
Throughout, we fix a filtered probability space $(\Omega,\mathcal{F},(\mathcal{F}_t)_{t \in \scr{T}},P)$, where either $\scr{T} = [0, T]$ for $T \in (0, \infty)$ (``finite time horizon") or $\scr{T} = [0, \infty)$ for $T = +\infty$ (``infinite time horizon"). To treat models with a finite and infinite time horizon in a unified manner, we fix a constant $\delta \geq 0$,\footnote{This will be the time-discount rate below; for infinite horizon models, it needs to be strictly positive.} and say that an $\mathbb{R}^\ell$-valued progressively measurable process $(X_t)_{t \in \scr{T}}$ belongs to $\scr{L}^{p}_\delta$, $p \geq 1$, if
$E[\int_0^Te^{-\delta t}\Vert X_t\Vert^p dt] < \infty$,
where $\Vert \cdot \Vert$ is any norm on $\mathbb{R}^\ell$. Likewise, an $\mathbb{R}^\ell$-valued local martingale $(M_t)_{t \in \scr{T}}$ belongs to $\scr{M}^{p}_\delta$, $p \geq 1$, if
$E[\Vert \int_0^T e^{-2 \delta s} d[M]_s \Vert^{p/2}] < \infty$. Here, $\Vert \cdot \Vert$ denotes any matrix norm on $\mathbb{R}^{\ell \times \ell}$.

\section{Model}\label{s:model}

\subsection{Financial Market}

We consider a financial market with $1+d$ assets. The first one is safe, and normalized to one for simplicity. The other $d$ assets are risky,  with dynamics driven by $d$-dimensional  Brownian motion  $(W_t)_{t \in \scr{T}}$:
\begin{equation}\label{eq:price}
dS_t=\mu_t dt+\sigma dW_t.
\end{equation}
Here, the $\mathbb{R}^{d}$-valued expected return process $(\mu_t)_{t\in \scr{T}} \in \scr{L}^2_\delta$ is to be determined in equilibrium, whereas the constant $\mathbb{R}^{d\times d}$-valued volatility matrix $\sigma$ is given exogenously. Throughout, we write $\Sigma= \sigma \sigma^\top$ and assume that this infinitesimal covariance matrix is nonsingular. 

\begin{remark}
Since our goal is to obtain a model with maximal tractability, it is natural to assume that the exogenous volatility matrix $\sigma$ is constant.\footnote{If one instead assumes that the volatility follows some (sufficiently integrable) stochastic process $(\sigma_t)_{t \in \scr{T}}$, then the subsequent characterization of individually optimal strategies and equilibrium returns in terms of coupled but linear FBSDEs as in  \eqref{eq:fbsde1}--\eqref{eq:fbsde2} still applies. However, the stochastic volatility then appears in the coefficients of this equation, so that the solution can no longer be characterized (semi-)explicitly in terms of matrix power series. Instead, a ``backward stochastic Riccati differential equation'' appears as a crucial new ingredient already in the one-dimensional models with exogenous price dynamics studied by \cite{kohlmann.tang.02,bank.voss.16}.} However, stochastic volatilities are bound to appear naturally in more general models where they are determined \emph{endogenously}. Such extensions of the current setting are an important direction for further research.
\end{remark}

\subsection{Endowments, Preferences, and Trading Costs}

A finite number of agents $n=1,\ldots,N$ receive (cumulative) random endowments $(Y^n_t)_{t \in \scr{T}}$ with dynamics
\begin{equation}\label{eq:endow}
dY^n_t=d A^n_t +(\zeta^n_t)^\top \sigma dW_t +dM_t^{\perp,n}, \quad n=1,\ldots,N.
\end{equation}
Here, the $\mathbb{R}$-valued adapted process $(A^n_t)_{t \in \scr{T}}$ with $E[\int_0^T e^{-\delta s} d \vert A \vert_s] < \infty$ denotes the finite variation component of Agent $n$'s endowment; it may contain lump-sum payments as well as absolutely continuous cash-flows. The $\mathbb{R}^{d}$-valued process $\zeta^n \in \scr{L}^2_\delta$ describes the exposure of the endowment to asset price shocks. Finally, the orthogonal $\mathbb{R}$-valued martingale $M^{\perp,n} \in \scr{M}^2_{\delta/2}$ models unhedgeable shocks. 

Without trading costs, the goal of Agent $n$ is to choose an $\mathbb{R}^{d}$-valued progressively measurable trading strategy $\varphi  \in \scr{L}^2_\delta$ (the number of shares held in each risky asset) to maximize the (discounted) expected changes of her wealth, penalized for the quadratic variation of wealth changes as in, e.g., \cite{garleanu.pedersen.13,garleanu.pedersen.16}:
\begin{align}
&E\left[\int_0^T e^{-\delta t}\left(\varphi _t^\top dS_t+dY^n_t- \frac{\gamma^n}{2}d\left\langle \textstyle\int_0^\cdot \varphi_s^\top dS_s+Y^n\right\rangle_t \right)\right] \notag\\
&=E\bigg[ \int_0^T e^{-\delta t} \left(\varphi_t^\top \mu_t - \frac{\gamma^n}{2} (\varphi_t+\zeta^n_t)^\top \Sigma  (\varphi_t+\zeta^n_t) \right)dt \notag \\
&\qquad\quad+ \int_0^T e^{-\delta t} \left(d A^n_t -\frac{\gamma^n}{2} d\langle M^{\perp,n} \rangle_t \right) \bigg] \to \max! \label{eq:frictionless}
\end{align}
Here, $\gamma^n>0$ and $\delta \geq 0$ are Agent $n$'s risk aversion and the (common) discount rate, respectively. We assume without loss of generality that
\begin{equation*}
\gamma^N = \max(\gamma^1, \ldots, \gamma^N),
\end{equation*}
so that Agent $N$ has the highest risk aversion among all agents. For simplicity, we also suppose that the initial stock position $\varphi^n_{0-}$ of each agent is zero.

\begin{remark}
A strictly positive discount rate allows to postpone the planning horizon indefinitely to obtain stationary infinite-horizon solutions as in \cite{martin.schoeneborn.11,martin.12,garleanu.pedersen.13,garleanu.pedersen.16}. In that case, $\varphi \in \scr{L}^2_\delta$ is an appropriate transversality condition that ensures that the problem is well posed. 
\end{remark}

The solution of the frictionless problem \eqref{eq:frictionless} is readily determined by pointwise optimization as
\begin{equation}\label{eq:flopt}
\varphi^n_t=\frac{\Sigma^{-1}\mu_t}{\gamma^n}-\zeta^n_t.
\end{equation}
The first term is the classical (myopic) Merton portfolio; the second is the mean-variance hedge for the replicable part of the endowment. 

As in \cite{almgren.chriss.01,grinold.06,garleanu.pedersen.13,cartea.jaimungal.16,garleanu.pedersen.16,almgren.li.16,guasoni.weber.15a,moreau.al.15,bank.voss.16,bank.al.17} we now assume that trades incur costs proportional to the square of the order flow $\dot{\varphi}_t=\frac{d}{dt}\varphi_t$.\footnote{The assumption of quadratic rather than proportional costs is made for tractability. However, buoyed by the results from the partial equilibrium literature, we expect the qualitative properties of our results to be robust across different \emph{small} transaction costs, compare with the discussion in \cite{moreau.al.15}.} This trading friction can either be interpreted as temporary price impact proportional to both trade size and trade speed, or as a (progressive) transaction tax or trading fee. For the first interpretation it is natural to assume that trades also move the prices of correlated securities (compare \cite{schied.al.10,garleanu.pedersen.13,garleanu.pedersen.16,guasoni.weber.15c}), and each agent's trades also affect the others' execution prices. In contrast, a tax as in \cite{subrahmanyam.98} or the fee charged by an exchange affects trades in each asset and by each agent separately. We focus on the second specification here, which simplifies the analysis by avoiding a coupling of the agents' optimization problems through common price impact. To wit, $\lambda^m > 0$, $m=1,\ldots,d$, describes the quadratic costs levied separately on each agent's order flow for asset $m$ and we denote by $\Lambda \in \mathbb{R}^{d \times d}$ the diagonal matrix with diagonal entries $\lambda^1, \ldots, \lambda^d$.\footnote{More general specifications do no seem natural for the tax interpretation of the model. Note, however, that the mathematical analysis below only uses that $\Lambda$ is symmetric and positive definite.} With this notation, Agent $n$'s optimization problem then reads as follows:

\begin{align}
J^n(\dot{\varphi}):=& E\left[ \int_0^T e^{-\delta t}\left(\varphi_t^\top \mu_t - \frac{\gamma^n}{2} (\varphi_t+\zeta^n_t)^\top \Sigma  (\varphi_t+\zeta^n_t) - \dot{\varphi}_t^\top \Lambda \dot{\varphi}_t dt\right)dt\right]\notag\\
&\quad +E\left[ \int_0^T  e^{-\delta t} \left(d A^n_t -\frac{\gamma^n}{2} d\langle M^{n\perp} \rangle_t\right) \right] \to \max!  \label{eq:goal}
\end{align}
In order to avoid infinite transaction costs, all trading rates (as well as the corresponding trading strategies themselves) naturally have to belong to $\scr{L}^2_\delta$.

The goal now is to solve for the equilibrium excess return that matches the agents' (and, potentially, noise traders') supply and demand. A similar model with a single strategic agent and noise traders with a particular parametric demand is analyzed in \cite[Section 4]{garleanu.pedersen.16}. Conversely, \cite{zitkovic.12,choi.larsen.15,kardaras.al.15,xinZit17} study models of the above form without noise traders (and with exponential rather than mean-variance preferences).

\section{Frictionless Equilibrium}\label{s:fless}

For later comparison to the frictional case, we first consider the model without trading costs. To clear the market, the expected return process $(\mu_t)_{t \in \scr{T}}$ needs to be chosen so that the demand of the strategic agents and the exogenous demand of a group of noise traders matches the total supply of zero at all times. To wit, modeling the noise trader demand by an exogenous process $\psi \in \scr{L}^2_\delta$ \texttt{with $\psi_{0-} = 0$}, the clearing condition reads as
$$0=\varphi^1_t+\ldots+\varphi^N_t+\psi_t.$$
(Alternatively, one can interpret $-\psi$ as the exogenous supply of the risky assets.) In view of \eqref{eq:flopt}, the frictionless equilibrium expected return therefore is
\begin{equation}\label{eq:mcapm}
\mu_t=\frac{\Sigma(\zeta^1_t+\zeta^2_t+\ldots+\zeta^N_t-\psi_t)}{1/\gamma^1+1/\gamma^2+\ldots+1/\gamma^N}.
\end{equation}
The interpretation is that the investment demand induced by the equilibrium return needs to offset the difference between the noise trading volume and the strategic agents' total hedging demand. Whence, the equilibrium return scales with the (exogenous) covariance matrix of the risky assets, relative to the total risk tolerance.

In this simple model, equilibrium dynamics and strategies are known in closed form, rather than only being characterized via martingale representation \cite{karatzas.shreve.98} or BSDEs~\cite{kardaras.al.15}. This makes the model an ideal point of departure for analyzing the impact of transaction costs on the equilibrium return.

\section{Individual Optimality with Transaction Costs}\label{s:indopt}

As a first step towards our general equilibrium analysis in Section~\ref{s:eqopt}, we now consider each Agent $n$'s individual optimization problem with transaction costs~\eqref{eq:goal}, taking an expected return process $\mu \in \scr{L}^2_\delta$ as exogenously given. A multidimensional generalization of the calculus of variations argument of \cite{bank.al.17} leads to the following representation of the optimal strategy in terms of a coupled but linear system of forward-backward stochastic differential equations (henceforth FBSDEs):

\begin{lemma}\label{thm:indopt1}
Let $\varphi^n_t =\frac{\Sigma^{-1}\mu_t}{\gamma^n}-\zeta^n_t$ be the frictionless optimizer from \eqref{eq:flopt}. Then the frictional optimization problem \eqref{eq:goal}  for Agent $n$ has a unique solution, characterized by the following FBSDE:
\begin{equation}\label{eq:FBSDE}
\begin{split}
d\varphi^{\Lambda, n}_t &=\dot{\varphi}^{\Lambda, n}_t dt, \quad \varphi^{\Lambda, n}_0 =0, \\
d\dot \varphi^{\Lambda, n}_t &= dM^n_t+\frac{\gamma^n \Lambda^{-1}\Sigma}{2}(\varphi^{\Lambda, n}_t-\varphi^n_t)dt + \delta \dot \varphi^{\Lambda, n}_t dt.
\end{split}
\end{equation}
Here, $\varphi^{\Lambda, n}_t, \dot \varphi^{\Lambda, n}_t \in \scr{L}^2_\delta$, and the $\mathbb{R}^d$-valued square-integrable martingale $M^n$ needs to be determined as part of the solution. If $T < \infty$, the dynamics~\eqref{eq:FBSDE} are complemented by the terminal condition\footnote{This means that agents stop trading near maturity, when there is not enough time left to recuperate the costs of further transactions. If $T=\infty$, this terminal condition is replaced by the transversality conditions implicit in $\varphi^{\Lambda, n}_t, \dot \varphi^{\Lambda, n}_t \in \scr{L}^2_\delta$ for $\delta>0$.}
\begin{equation}
\label{eq:FBSDE:terminal condition}
\dot \varphi^{\Lambda, n}_T = 0.
\end{equation}
For $T=\infty$, Agent $n$'s unique individually optimal strategy $\varphi^{\Lambda,n}$ has the explicit representation \eqref{eq:varphi:infinite}; the corresponding optimal trading rate $\dot{\varphi}^{\Lambda,n}$ is given in feedback form by \eqref{eq:ODE:infinite}. For $T<\infty$, the corresponding formulas are provided in \eqref{eq:varphi} and~\eqref{eq:ODE}, respectively.
\end{lemma}

\begin{proof}
Since the goal functional~\eqref{eq:goal} is strictly convex, \eqref{eq:goal}  has a unique solution if and only if there exists a (unique) solution to the following first-order condition~\cite{ekeland.temam.99}:
\begin{equation}\label{eq:gateaux}
\left\langle J'\left(\dot \varphi \right),\dot{\vartheta}\right\rangle=0, \quad \mbox{for all } \vartheta \text{ with } \vartheta_0 = 0 \text{ and } \vartheta, \dot \vartheta \in \scr{L}^2_\delta.
\end{equation}
Here, the G\^{a}teaux derivative of $J$ in the direction $\dot{\vartheta}$ is given by
\begin{align*}
\left\langle J'\left(\dot \varphi \right),\dot{\vartheta}\right\rangle &= \lim_{\rho \to 0} \frac{J(\dot \varphi +\rho\dot{\vartheta})-J(\dot \varphi)}{\rho}\\
&=E\left[\int_0^T e^{-\delta t}\left(  (\mu^\top_t-\gamma^n  (\varphi_t+\zeta^n_t)^\top \Sigma) \left(\int_0^t \dot{\vartheta}_s ds\right)- 2\left(\dot \varphi_t\right)^\top \Lambda \dot{\vartheta}_t\right)  dt\right].
\end{align*}
By Fubini's theorem,
\begin{align*}
\int_0^T  &\left(e^{-\delta t}(\mu^\top_t -\gamma^n\left(\varphi_t+\zeta^n_t\right)^\top \Sigma)  \left(\int_0^t \dot{\vartheta}_s ds \right)\right)dt \\
&=\int_0^T \left(\int_s^T e^{-\delta t}\left(\mu^\top_t -\gamma^n\left(\varphi_t+\zeta^n_t\right)^\top \Sigma\right) dt \right)  \dot{\vartheta}_s ds.
\end{align*}
Together with the tower property of the conditional expectation, this allows to rewrite the first-order condition \eqref{eq:gateaux} as
\begin{equation*}
0= E\left[\int_0^T \left( E\left[\int_t^T e^{-\delta s}\left(\mu^\top_s -\gamma^n\left(\varphi_s+\zeta^n_s\right)^\top \Sigma\right)ds  \Big|\mathcal{F}_t\right] - 2e^{-\delta t} \left(\dot\varphi \right)^\top \Lambda \right)  \dot{\vartheta}_t dt\right].
\end{equation*}
Since this has to hold for any perturbation $\dot{\vartheta}$, ~\eqref{eq:goal} has a (unique) solution $\dot \varphi^{\Lambda, n}$ if and only if 
\begin{align}\label{eq:rep}
\dot \varphi^{\Lambda, n}_t &= \frac{\gamma^n \Lambda^{-1}\Sigma}{2} e^{\delta t} E\left[\int_t^T e^{-\delta s}\left(\frac{\Sigma^{-1}\mu_s}{\gamma^n} -\zeta^n_s -\varphi^{\Lambda, n}_s\right)ds \Big|\mathcal{F}_t\right]
\end{align}
has a a (unique) solution. 

Now, assume that \eqref{eq:rep} has a (unique) solution $\dot \varphi^{\Lambda, n}$. Note that if $T < \infty$, \eqref{eq:FBSDE:terminal condition} is satisfied. Define the square-integrable martingale $\tilde{M}_t=\frac{\gamma^n \Lambda^{-1} \Sigma}{2} E\big[\int_0^T e^{-\delta s}(\varphi^n_s -\varphi^{\Lambda, n}_s)ds|\mathcal{F}_t\big]$, $t \in \scr{T}$. Integration by parts then allows to rewrite \eqref{eq:rep} as
$$d\dot{\varphi}^{\Lambda,n}_t =e^{\delta t} d\tilde{M}_t- \frac{\gamma^n\Lambda^{-1}\Sigma}{2} (\varphi^n_t-\varphi^{\Lambda, n}_t)dt+\delta \dot \varphi^{\Lambda, n}_t dt.$$
Together with the definition $d\varphi_t^{\Lambda, n}= \dot{\varphi}_t^{\Lambda, n} dt$, this yields the claimed FBSDE representation \eqref{eq:FBSDE}.

Conversely, assume that \eqref{eq:FBSDE} has a (unique) solution $(\varphi^{\Lambda, n}, \dot \varphi^{\Lambda, n}, M^n)$, where $\varphi^{\Lambda, n}, \dot \varphi^{\Lambda, n} \in \scr{L}^2_\delta$ and $M^n$ is an $\mathbb{R}^\ell$-valued martingale with finite second moments. 

 First note that, for $t \in \scr{T}$ with $t < \infty$, integration by parts gives
\begin{equation}
\label{eq:pf:thm:indopt1:int parts}
e^{-\delta t} \dot \varphi^{\Lambda, n}_t = \dot \varphi^{\Lambda, n}_0 + \int_0^t e^{-\delta s} dM^n_s + \int_0^t e^{-\delta s}\frac{\gamma^n \Lambda^{-1}\Sigma}{2}(\varphi^{\Lambda, n}_s-\varphi^n_t)ds.
\end{equation}
Next, we claim that
\begin{equation}
\label{eq:pf:thm:indopt1:claim}
\dot \varphi^{\Lambda, n}_0 =  - \int_0^T e^{-\delta s} dM^n_s - \int_0^T e^{-\delta s}\frac{\gamma^n \Lambda^{-1}\Sigma}{2}(\varphi^{\Lambda, n}_s-\varphi^n_t)ds.
\end{equation}
If $T < \infty$, this follows from \eqref{eq:pf:thm:indopt1:int parts} for $t = T$ together with the terminal condition \eqref{eq:FBSDE:terminal condition}. If $T = \infty$, we argue as follows: since $\dot \varphi^{\Lambda, n} \in \scr{L}^2_\delta$, there exists an increasing sequence $(t_k)_{k \in \mathbb{N}}$ with $\lim_{k \to \infty} t_k = \infty$ along which the left-hand side of \eqref{eq:pf:thm:indopt1:int parts} converges a.s.~to zero. Moreover, Proposition \ref{prop:BSDE:martingale}, the martingale convergence theorem and $\varphi^{\Lambda, n} ,\varphi^n \in \scr{L}^2_\delta$ show that the right-hand side of \eqref{eq:pf:thm:indopt1:int parts} converges (along $t_k$) a.s. to 
$$\dot \varphi^{\Lambda, n}_0 + \int_0^\infty e^{-\delta s} dM^n_s + \int_0^\infty e^{-\delta s}\frac{\gamma^n \Lambda^{-1}\Sigma}{2}(\varphi^{\Lambda, n}_s-\varphi^n_t)ds.$$
Hence, \eqref{eq:pf:thm:indopt1:claim} holds also in this case.

Inserting \eqref{eq:pf:thm:indopt1:claim} into \eqref{eq:pf:thm:indopt1:int parts}, taking conditional expectations and rearranging in turn yields \eqref{eq:rep}.

It remains to show that the FBSDE \eqref{eq:FBSDE} has a (unique) solution $(\varphi^{\Lambda, n}, \dot \varphi^{\Lambda, n}, M^n)$. Since the matrix $\frac{\gamma^n}{2} \Lambda^{-1}\Sigma$ has only positive eigenvalues (because it is the product of two symmetric positive definite matrices, cf.~\cite[Proposition 6.1]{serre.10}), this follows from Theorem~\ref{thm:fbsde:infinite} (for $T=\infty$) or Theorem~\ref{thm:fbsde} (for $T<\infty$), respectively.
\end{proof}

\section{Equilibrium with Transaction Costs}\label{s:eqopt}

\subsection{Equilibrium Returns}

We now use the above characterization of individually optimal strategies to determine the \emph{equilibrium return} $(\mu_t)_{t \in \scr{T}}$, for which the agents' individually optimal demands match the zero net supply of the risky asset at all times. As each agent's trading rate is now constrained to be absolutely continuous, the same needs to hold for the exogenous noise-trading volume: 
$$
d\psi_t=\dot{\psi}_tdt,
$$ 
where $d\dot{\psi}_t=\mu^\psi_tdt+ dM_t^\psi$ for $\mu^\psi \in \scr{L}^2_\delta$ and a local martingale $M^\psi$. We also assume that $\psi, \dot \psi \in \scr{L}^2_\delta$. The key ingredient for the equilibrium return is the solution of another system of coupled but linear FBSDEs:

\begin{lemma}\label{lem:FBSDE2}
There exists a unique solution  $(\varphi^\Lambda,\dot \varphi^\Lambda)=(\varphi^{\Lambda,1},\ldots,\varphi^{\Lambda, N-1},\dot{\varphi}^{\Lambda,1},\ldots,\dot{\varphi}^{\Lambda,N-1})$ of the following FBSDE:
\begin{align}
\label{eq:FBSDE3}
\begin{split}
d\varphi^\Lambda_t &=\dot{\varphi}^\Lambda_t dt, \quad \varphi_0=0,\\
d\dot{\varphi}^\Lambda_t&=dM_t +\left(B\varphi^\Lambda_t+\delta\dot{\varphi}^\Lambda_t-A\zeta_t+\chi_t\right) dt,
\end{split}
\end{align}
satisfying the terminal condition $\dot \varphi^\Lambda_T = 0$ if $T < \infty$. Here, $M$ is an $\mathbb{R}^{d(N-1)}$-valued martingale with finite second moments, $\zeta=((\zeta^1)^\top,\ldots,(\zeta^N)^\top)^\top$,
\begin{align*}
B&= \begin{pmatrix} \left(\frac{\gamma^N-\gamma^1}{N}+\gamma^1\right) \frac{\Lambda^{-1}\Sigma}{2} & \cdots &\frac{\gamma^{N}-\gamma^{N-1}}{N} \frac{\Lambda^{-1}\Sigma}{2}  \\ \vdots & \ddots & \vdots  \\ \frac{\gamma^N-\gamma^1}{N} \frac{\Lambda^{-1}\Sigma}{2} & \cdots & \left(\frac{\gamma^{N}-\gamma^{N-1}}{N}+\gamma^{N-1} \right) \frac{\Lambda^{-1}\Sigma}{2} \end{pmatrix} \in \mathbb{R}^{d (N-1)\times d (N-1)},\\
A&=\begin{pmatrix} \left(\frac{\gamma^1}{N}-\gamma^1\right) \frac{\Lambda^{-1}\Sigma}{2} & \cdots &\frac{\gamma^{N-1}}{N} \frac{\Lambda^{-1}\Sigma}{2} & \frac{\gamma^N}{N} \frac{\Lambda^{-1}\Sigma}{2} \\ \vdots & \ddots & \vdots & \vdots \\ \frac{\gamma^1}{N} \frac{\Lambda^{-1}\Sigma}{2} & \cdots & \left(\frac{\gamma^{N-1}}{N}-\gamma^{N-1}\right) \frac{\Lambda^{-1}\Sigma}{2} & \frac{\gamma^N}{N} \frac{\Lambda^{-1}\Sigma}{2} \end{pmatrix} \in \mathbb{R}^{d(N-1)\times d N},
\end{align*}
and
$$
\chi_t=\frac{1}{N}\left(\left(\frac{\gamma^N\Lambda^{-1}\Sigma}{2}\psi_t+\delta\dot{\psi}_t-\mu^\psi_t\right)^\top,\ldots,\left(\frac{\gamma^N\Lambda^{-1}\Sigma}{2}\psi_t+\delta\dot{\psi}_t-\mu^\psi_t\right)^\top\right)^\top \in \mathbb{R}^{d(N-1)}.
$$
\end{lemma}

\begin{proof}
Lemma \ref{lem:eigen} shows that all eigenvalues of the matrix $B$ are real and positive; in particular, $B$ is invertible. The assertion in turn follows from Theorem \ref{thm:fbsde:infinite} for $T = \infty$ and from Theorem \ref{thm:fbsde} for $T < \infty$ because $\zeta, \chi \in \scr{L}^2_\delta$.
\end{proof}

We can now state our main result:

\begin{theorem}\label{thm:main}
The unique frictional equilibrium return is 
\begin{equation}\label{eq:mu}
\mu^\Lambda_t =\sum_{n=1}^{N-1} \frac{(\gamma^n-\gamma^N)\Sigma}{N} \varphi^{\Lambda,n}_t+\sum_{n=1}^N \frac{\gamma^n\Sigma}{N} \zeta^n_t-\frac{\gamma_N \Sigma}{N}\psi_t+ \frac{2\Lambda}{N}(\mu^\psi_t-\delta\dot{\psi}_t).
\end{equation}
The corresponding individually optimal trading strategies of Agents $n=1,\ldots,N$ are $\varphi^{\Lambda,1},\ldots,\varphi^{\Lambda, N-1}$ from Lemma~\ref{lem:FBSDE2} and $\varphi^{\Lambda,N}=-\sum_{n=1}^{N-1}\varphi^{\Lambda,n}-\psi$.
\end{theorem}

\begin{proof}
Let $\nu \in \scr{L}^2_\delta$ be \emph{any} equilibrium return and denote by $\vartheta^\Lambda=(\vartheta^{\Lambda,1},\ldots,\vartheta^{\Lambda,N})$ the corresponding individually optimal trading strategies. Then, market clearing implies that not only the positions of the agents but also their trading rates must sum to zero, $0=\sum_{n=1}^N \dot{\vartheta}^{\Lambda,n}+\dot{\psi}$. Together with the FBSDEs \eqref{eq:FBSDE} describing each agent's optimal trading rate, it follows that
$$
0=dM_t +\sum_{n=1}^N \frac{\Lambda^{-1}}{2} \big(\gamma^n \Sigma \vartheta^{\Lambda,n}_t -(\nu_t -\gamma^n \Sigma \zeta^n_t) \big)dt+\sum_{n=1}^N \delta \dot{\vartheta}^{\Lambda,n}_t dt+d\dot{\psi}_t,
$$
for a local martingale $M$. Market clearing implies $\vartheta^{\Lambda,N}=-\sum_{n=1}^{N-1} \vartheta^{\Lambda,n}-\psi$, and so this gives
$$
0=dM_t + \frac{\Lambda^{-1}}{2} \left(\sum_{n=1}^{N-1}(\gamma^n-\gamma^N) \Sigma \vartheta^{\Lambda,n}_t -\sum_{n=1}^N (\nu_t -\gamma^n \Sigma \zeta^n_t)-\gamma_N \Sigma \psi_t \right)dt-\delta \dot{\psi}_t dt+\mu_t^\psi dt + dM^\psi_t.
$$
Since any continuous local martingale of finite variation is constant, it follows that
\begin{equation}\label{eq:drift}
\nu_t =\sum_{n=1}^{N-1} \frac{(\gamma^n-\gamma^N)\Sigma}{N} \vartheta^{\Lambda,n}_t+\sum_{n=1}^N \frac{\gamma^n\Sigma}{N} \zeta^n_t-\frac{\gamma_N \Sigma}{N}\psi_t+ \frac{2\Lambda}{N}(\mu^\psi_t-\delta\dot{\psi}_t).
\end{equation}
Plugging this expression for $\nu_t$ back into Agent $n=1,\ldots,N-1$'s individual optimality condition~\eqref{eq:FBSDE}, we deduce that
\begin{align*}
d\dot{\vartheta}^{\Lambda,n}_t = dM^n_t &+\frac{\Lambda^{-1}\Sigma}{2}\left(\gamma^n\vartheta^{\Lambda,n}_t +\sum_{m=1}^{N-1}\frac{\gamma^N-\gamma^m}{N} \vartheta^{\Lambda,m}_t+\gamma^n\zeta_t^n -\sum_{m=1}^N \frac{\gamma^m}{N}\zeta^m_t\right)dt\\
&+\frac{1}{N} \left(\frac{\gamma_N\Lambda^{-1}\Sigma}{2}\psi_t+\delta\dot{\psi}_t-\mu^\psi_t\right)dt, \qquad n=1,\ldots,N-1.
\end{align*}
Hence, $(\vartheta^{\Lambda,1},\ldots,\vartheta^{\Lambda,N-1},\dot{\vartheta}^{\Lambda,1},\ldots,\dot{\vartheta}^{\Lambda,N-1})$ solves the FBSDE~\eqref{eq:FBSDE3} and therefore coincides with its \emph{unique} solution from Lemma~\ref{lem:FBSDE2}. Market clearing in turn shows $\vartheta^{\Lambda,N}=\varphi^{\Lambda,N}$, and \eqref{eq:drift} implies that the equilibrium return coincides with \eqref{eq:mu}. This establishes that if an equilibrium exists, then it has to be of the proposed form.

To verify that the proposed returns process and trading strategies indeed form an equilibrium, we revert the above arguments. Market clearing holds by definition of $\varphi^{\Lambda,N}$, so it remains to check that $\varphi^{\Lambda,n}$ is indeed optimal for agent $n=1,\ldots,N$. To this end, it suffices to show that the individual optimality conditions~\eqref{eq:rep} are satisfied for $n=1,\ldots,N$. After inserting the definitions of $\mu^\Lambda$, one first realises that for $n =1, \ldots, N-1$, \eqref{eq:rep} coincides with the respective equation in \eqref{eq:FBSDE3}, and for $n=N$, this follows from market clearing. This completes the proof.
\end{proof}

\subsection{Equilibrium Liquidity Premia}

Let us now discuss the \emph{equilibrium liquidity premia} implied by Theorem~\ref{thm:main}, i.e., the differences between the frictional equilibrium returns \eqref{eq:mu} and their frictionless counterparts \eqref{eq:mcapm}.
To this end, denote by $\bar \varphi^n$, $n=1,\ldots,N$,  the frictionless optimal strategy from \eqref{eq:flopt} for Agent $n$, corresponding to the frictionless equilibrium return~\eqref{eq:mcapm}:
\begin{equation*}
\bar \varphi^n_t = \frac{1/\gamma_n \left(\sum_{m=1}^N
	\zeta^m_t - \psi_t \right)}{\sum_{m=1}^N 1/\gamma^m}  - \zeta^n_t.
\end{equation*}
With this notation, the frictionless equilibrium return $\mu$ can be written as
\begin{equation}
\label{eq:mu frictionless rewritten}
\mu_t = \sum_{n=1}^N \frac{\gamma_n \Sigma}{N} (\bar \varphi^n_t + \zeta^n_t).
\end{equation}
Now subtract~\eqref{eq:mu frictionless rewritten} from the frictional equilibrium return \eqref{eq:mu}, use that $- \sum_{n =1}^{N-1} \varphi^{\Lambda, n} = \varphi^{\Lambda_N} + \psi_t$ by the frictional clearing condition, and note that $\sum_{n=1}^N ( \varphi_t^{\Lambda,n}  - \bar \varphi_t^n) = (-\psi_t + \psi_t) = 0$ by frictional and frictionless market clearing. This yields the following expression for the liquidity premium:
\begin{align}
\mathrm{LiPr}_t := \mu^\Lambda_t - \mu_t &= \sum_{n=1}^{N} \frac{\gamma^n\Sigma}{N} \varphi^{\Lambda,n}_t+\sum_{n=1}^N \frac{\gamma^n\Sigma}{N} \zeta^n_t + \frac{2\Lambda}{N}(\mu^\psi_t-\delta\dot{\psi}_t) - \sum_{n=1}^N \frac{\gamma_n \Sigma}{N} (\bar \varphi^n_t + \zeta_t) \notag \\
&=\frac{\Sigma}{N} \sum_{n=1}^N \gamma^n ( \varphi_t^{\Lambda,n}  - \bar \varphi_t^n) + \frac{2\Lambda}{N}(\mu^\psi_t-\delta\dot{\psi}_t) \notag \\
&=\frac{\Sigma}{N} \sum_{n=1}^N (\gamma^n -\bar \gamma)( \varphi_t^{\Lambda,n}  - \bar \varphi_t^n) + \frac{2\Lambda}{N}(\mu^\psi_t-\delta\dot{\psi}_t), 
\label{eq:LiPrM}
\end{align}
where $\bar \gamma = \sum_{n =1}^N \frac{\gamma_n}{N}$ denotes the average risk aversion of the strategic agents.

Let us now interpret this result. A first observation is that if all agents are strategic and have the same risk aversion, then the frictionless equilibrium returns \eqref{eq:mcapm} also clear the market with transaction costs:

\begin{corollary}\label{cor:homo}
Suppose there are no noise traders and all strategic agents have the same risk aversion $\bar \gamma=\gamma^1=\ldots= \gamma^N$. Then there are no liquidity premia.
\end{corollary}

A similar result has been established for exponential investors in the limit for small transaction costs by \cite{herdegen.muhlekarbe.17}. In the present quadratic context, this result holds true exactly. A result in the same spirit in a static model is \cite[Corollary 4.12]{anthropelos.17}, where incompleteness also only affects strategies but not equilibrium prices for mean-variance investors with homogenous risk aversions. Another related result is \cite[Proposition 12]{sannikov.skrzypacz.16}, where homogeneous risk aversion imply that the frictional equilibrium converges to the frictionless one as the horizon grows. 

However, this result no longer remains true in the presence of noise traders:
\begin{corollary}\label{cor:homo2}
	Suppose that all strategic agents have the same risk aversion $\bar \gamma=\gamma^1=\ldots= \gamma^N$. Then:
	\begin{equation*}
	\mathrm{LiPr}_t = \frac{2\Lambda}{N}(\mu^\psi_t-\delta \dot{\psi}_t).
	\end{equation*}
\end{corollary}
To illustrate the intuition behind this result, consider the simplest case where the noise traders simply sell at a constant rate, $\dot{\psi}<0$. Put differently, the number of risky shares available for trading expands linearly. Then $\mathrm{LiPr}_t  = -\frac{2\Lambda}{N}\delta \dot{\psi}>0$. This illustrates how market growth can lead to positive liquidity premia even for homogenous agents. 

If there is only one strategy agent ($N =1$), we are always in the setting of Corollary~\ref{cor:homo2}. An example is the model of Garleanu and Pedersen~\cite[Section 4]{garleanu.pedersen.16} with a single risky asset ($d=1$), a single strategic agent without random endowment ($\zeta=0$) and exogenous noise traders, whose positions $\psi_t$ are mean-reverting around a stochastic mean:\footnote{Several groups of noise traders with different mean positions as considered in \cite[Section 4]{garleanu.pedersen.16} can be treated analogously.} 
	$$d\psi_t=\kappa_\psi(X_t-\psi_t)dt,$$
	where $X$ is an Ornstein-Uhlenbeck process driven by a Brownian motion $W^X$:
	$$dX_t=-\kappa_X X_t dt+\sigma_X dW^X_t.$$
	In the notation from Section~\ref{s:eqopt}, we then have $\dot{\psi}_t=\kappa_\psi (X_t-\psi_t)$ and
	$$
	\mu^\psi_t=-\kappa_\psi \kappa_X X_t - \kappa_\psi \dot{\psi}_t=-\kappa_\psi \kappa_X X_t - \kappa^2_\psi (X_t-\psi_t) = \kappa_\psi^2 \psi_t -\kappa_\psi(\kappa_\psi+\kappa_X)X_t.
	$$  
We therefore recover the nontrivial liquidity premia of \cite[Proposition 9]{garleanu.pedersen.16}, to which we also refer for a discussion of the corresponding comparative statics.

\medskip

To obtain nontrivial liquidity premia in a model with only strategic agents and a fixed supply of risky assets, one needs to consider agents with heterogeneous risk aversions. To ease notation and interpretation, suppose there are no noise traders ($\psi=0$). Then, the liquidity premium~\eqref{eq:LiPrM} is simplifies to
\begin{equation*}
 \mathrm{LiPr}_t= \frac{\Sigma}{N} \sum_{n=1}^N (\gamma^n - \bar \gamma)( \varphi_t^{\Lambda,n}  - \bar \varphi_t^n).
\end{equation*}
This means that the liquidity premium is the sample covariance between the vector $(\gamma^1,\ldots,\gamma^N)$ of risk aversions and the current deviations $(\varphi^{\Lambda,1}_t-\bar \varphi^1_t,\ldots,\varphi^{\Lambda,N}_t-\bar \varphi^N_t)$ between the agents' actual positions and their frictionless targets. Hence, the liquidity premium is positive if and only if sensitivity and excess exposure to risk are positively correlated, i.e., if the more risk averse agents hold larger risky positions than in the (efficient) frictionless equilibrium. Then, these agents will tend to be net sellers and, as their trading motive is stronger than for the net buyers, a positive liquidity premium is needed to clear the market. 
  
  \medskip
  
  To shed further light on the dynamics of liquidity premia induced by heterogenous risk aversions, we now consider some concrete examples where the aggregate of the agents' endowments is zero as in~\cite{lo.al.04}. First, we consider the simplest case where endowment exposures have independent, stationary increments: 

\begin{corollary}\label{cor:OU 1}
	Let the time horizon be infinite and consider two strategic agents with risk aversions $\gamma_1 < \gamma_2$, discount rate $\delta>0$, and endowment volatilities following arithmetic Brownian motions:
	\begin{equation*}\label{eq:zeta}	
	\zeta^1_t=at+N_t, \qquad \zeta^2_t =-\zeta^1_t,
	\end{equation*}
	for an $\mathbb{R}^d$-valued Brownian motion $N$ and $a \in \mathbb{R}^d$. Then, the frictionless equilibrium return vanishes. The equilibrium return with transaction costs has Ornstein-Uhlenbeck dynamics:
	$$
	d\mu^\Lambda_t =\left(\sqrt{\frac{\gamma_1+\gamma_2}{2}\frac{\Sigma \Lambda^{-1}}{2}+\frac{\delta^2}{4}I_d}-\frac{\delta}{2}I_d\right) \left(2\frac{\gamma_1-\gamma_2}{\gamma_1+\gamma_2}\delta \Lambda a-\mu^\Lambda_t\right)dt +\frac{(\gamma_1-\gamma_2)\Sigma}{2}dN_t.
	$$
\end{corollary}

\begin{proof}
	The frictionless equilibrium return vanishes by \eqref{eq:mcapm}. In view of Theorem~\ref{thm:main} and since $\zeta^2 =-\zeta^1$, its frictional counterpart is given by
	\begin{equation}\label{eq:eqmu}
	\mu^\Lambda_t= \frac{(\gamma^1-\gamma^2)\Sigma}{2} (\varphi^{\Lambda,1}_t+\zeta^1_t).
	\end{equation}
	By  Lemma~\ref{lem:FBSDE2} as well as Theorem~\ref{thm:fbsde:infinite} (with $B=\frac{\gamma_1+\gamma_2}{2}\frac{\Lambda^{-1}\Sigma}{2}$ and $\xi_t=-\zeta^1_t$) and the representation \eqref{eq:ODE:infinite} from its proof, Agent 1's optimal trading rate is 
	\begin{equation}
	\label{eq:ex:trading rate}
	\dot{\varphi}^{\Lambda,1}_t=\bar{\xi}^1_t- \left(\sqrt{\Delta}- \frac{\delta}{2}I_d\right)\varphi^{\Lambda,1}_t,
	\end{equation}
	where
	\begin{align*}
\Delta &= \frac{\gamma_1+\gamma_2}{2}\frac{\Lambda^{-1}\Sigma}{2}+\frac{\delta^2}{4}I_d, \\
	\bar{\xi}^1_t &= -\left(\sqrt{\Delta} -\tfrac{\delta}{2}I_d \right) E\left[\int_t^\infty \left(\sqrt{\Delta} +\tfrac{\delta}{2} I_d \right) e^{-(\sqrt{\Delta}+\frac{\delta}{2}I_d)(u-t)}(a u+N_u) d u\Big|\mathcal{F}_t\right] \\
	&= -\left(\sqrt{\Delta} - \frac{\delta}{2}I_d\right) \left(\zeta^1_t+(\sqrt{\Delta}+\tfrac{\delta}{2}I_d)^{-1}a\right).
	\end{align*}
	Here, we have used an elementary integration and the martingale property of $N$ for the last equality. Plugging this back into \eqref{eq:ex:trading rate} yields
	\begin{equation*}
	\dot{\varphi}^{\Lambda,1}_t=-\left(\sqrt{\Delta} - \frac{\delta}{2}I_d\right)\left(\varphi^{\Lambda,1}_t + \zeta^1_t+(\sqrt{\Delta} + \tfrac{\delta}{2}I_d)^{-1}a\right).
	\end{equation*}
	Inserting this into \eqref{eq:eqmu} in turn leads to the asserted Ornstein-Uhlenbeck dynamics:
	\begin{align*}
	d \mu^\Lambda_t &= \tfrac{(\gamma^1-\gamma^2)\Sigma}{2} (\dot \varphi^{\Lambda,1}_t dt+d\zeta^1_t) \\
	&= \tfrac{(\gamma^1-\gamma^2)\Sigma}{2} \left(\left(-\left(\sqrt{\Delta} - \tfrac{\delta}{2}I_d\right)\Big(\varphi^{\Lambda,1}_t + \zeta^1_t+(\sqrt{\Delta}+\tfrac{\delta}{2}I_d)^{-1}a\Big) + a\right) dt+d N_t\right) \\
	&= \left(\tfrac{(\gamma^1-\gamma^2)\Sigma}{2}  \left(\sqrt{\Delta} - \tfrac{\delta}{2}I_d\right) (\Delta -\tfrac{\delta^2}{4}I_d)^{-1}\delta a) -  \Sigma \left(\sqrt{\Delta} - \tfrac{\delta}{2}I_d\right) \Sigma^{-1}\mu_t^\Lambda\right) dt + \tfrac{(\gamma^1-\gamma^2)\Sigma}{2} dN_t \\
	&= \left(2\tfrac{\gamma^1-\gamma^2}{\gamma_1+\gamma_2}  \Sigma \left(\sqrt{\Delta} - \tfrac{\delta}{2}I_d\right) \Sigma^{-1} \delta \Lambda a -  \Sigma \left(\sqrt{\Delta} - \tfrac{\delta}{2}I_d\right) \Sigma^{-1}\mu_t^\Lambda\right) dt + \tfrac{(\gamma^1-\gamma^2)\Sigma}{2} dN_t \\
	&= \Sigma \left(\sqrt{\Delta} - \tfrac{\delta}{2}I_d\right) \Sigma^{-1} \left(2\tfrac{\gamma^1-\gamma^2}{\gamma_1+\gamma_2}\delta \Lambda a- \mu^\Lambda_t\right) dt + \tfrac{(\gamma^1-\gamma^2)\Sigma}{2} dN_t \\
	&= \left(\sqrt{\tfrac{\gamma_1+\gamma_2}{2}\tfrac{\Sigma \Lambda^{-1}}{2}+\tfrac{\delta^2}{4}I_d}-\tfrac{\delta}{2}I_d\right) \left(2\tfrac{\gamma_1-\gamma_2}{\gamma_1+\gamma_2}\delta \Lambda a-\mu^\Lambda_t\right)dt +\tfrac{(\gamma_1-\gamma_2)\Sigma}{2}dN_t,
	\end{align*}
	where we have used Lemma \ref{lem:matrix function:properties}(b) in the last equality.
	\end{proof}

Let us briefly discuss the comparative statics of the above formula. In line with Corollary~\ref{cor:homo}, the average liquidity premium $2\frac{\gamma_1-\gamma_2}{\gamma_1+\gamma_2}\delta \Lambda a$ and the corresponding volatility both vanish if the agents' risk aversions coincide. More generally, its size is proportional to the degree of heterogenity, measured by $\frac{\gamma_1-\gamma_2}{\gamma_1+\gamma_2}$, multiplied by the discount rate $\delta$, the trading cost $\Lambda$, and the trend $a$ of Agent $1$'s position. To understand the intuition behind this result, suppose that $a < 0$ so that Agent 1 is a net buyer and Agent 2 is a net seller. Since $\gamma_1 < \gamma_2$, Agent's 2 motive to sell dominates Agent's 1 motive to buy and hence an additional positive drift is required to clear the market with friction. Since these readjustments do not happen immediately but only gradually over time, the size of these effects is multiplied by the discount rate $\delta$: a higher demand for immediacy forces the more risk averse agent to pay a larger premium.

Even for correlated assets, the \emph{average} liquidity premium only depends on the trading cost for the respective asset and the corresponding imbalance in Agent 1 and 2's demands.
The linear scaling in the trading cost also shows that for small costs, the average value of the liquidity premium is much smaller than its standard deviations (which then scale with the square-root of the liquidity premium).

Finally, note that liquidity premia are mean reverting here even though the agents' endowments are not stationary. The reason is the sluggishness of the agents' portfolios with transaction costs: the stronger trading need of the more risk-averse agent is not realized immediately but only gradually. This endogenously leads to autocorrelated returns like in the reduced-form models from the frictionless portfolio choice literature~\cite{kim.omberg.96,wachter.02}. 

\medskip 

Next, we turn to a stationary model where endowment exposures are also mean-reverting as in \cite{christensen.larsen.14,garleanu.pedersen.16}. This generates an additional state variable, that appears in the corresponding liquidity premia as a stochastic mean-reversion level:

\begin{corollary}\label{cor:OU 2}
	Let the time horizon be infinite and consider two strategic agents with risk aversions $\gamma_1 < \gamma_2$, discount rate $\delta>0$, and endowment volatilities with Ornstein-Uhlenbeck dynamics:
	\begin{equation*}\label{eq:zeta2}	
	d\zeta^1_t=-\kappa \zeta^1_t dt  +dN_t, \qquad d\zeta^2_t =-d\zeta^1_t, \quad \zeta^1_0=\zeta^2_0=0,
	\end{equation*}
	for an $\mathbb{R}^d$-valued Brownian motion $N$ and a positive-definite mean-reversion matrix $\kappa \in \mathbb{R}^{d \times d}$. Then, the frictionless equilibrium return vanishes. The equilibrium return with transaction costs has Ornstein-Uhlenbeck-type dynamics with a stochastic mean-reversion level that is a constant multiple of the endowment levels:
\begin{align*}
	d\mu^\Lambda_t &=\Sigma \left(\sqrt{\Delta} - \frac{\delta}{2}I_d\right) \Sigma^{-1} \left(\frac{(\gamma^1-\gamma^2) \Sigma}{2} \Big(\kappa (\sqrt{\Delta}+\tfrac{\delta}{2} I_d +\kappa)^{-1}-(\sqrt{\Delta}-\tfrac{\delta}{2} I_d)^{-1}\kappa\Big)\zeta^1_t-\mu^\Lambda_t\right)dt  \\
	&\quad+\frac{(\gamma^1-\gamma^2)\Sigma}{2}dN_t,
\end{align*}
where 
\begin{equation*}
\Delta = \frac{\gamma_1+\gamma_2}{2}\frac{\Lambda^{-1}\Sigma}{2}+\frac{\delta^2}{4}I_d.
\end{equation*}
\end{corollary}

\begin{proof}
	As in Corollary~\ref{cor:OU 1}, the frictionless equilibrium return vanishes by \eqref{eq:mcapm}, its frictional counterpart is given by~\eqref{eq:eqmu}, and Agent 1's optimal trading rate is \eqref{eq:ex:trading rate}. The only change is the target process $\bar{\xi}^1$, which can be computed as follows in the present context:
	\begin{align*}
	\bar{\xi}^1_t &= -\left(\sqrt{\Delta} -\tfrac{\delta}{2}I_d \right) \int_t^\infty \left(\sqrt{\Delta} +\tfrac{\delta}{2} I_d \right) e^{-(\sqrt{\Delta}+\frac{\delta}{2}I_d)(u-t)}E[\zeta^1_u|\mathcal{F}_t] du \\
	&= -\left(\sqrt{\Delta} - \tfrac{\delta}{2}I_d\right) (\sqrt{\Delta}+\tfrac{\delta}{2}I_d) (\sqrt{\Delta}+\tfrac{\delta}{2}I_d+\kappa)^{-1} \zeta^1_t.
	\end{align*}
	Here, we have used the expectation $E[\zeta^1_u|\mathcal{F}_t]=e^{-\kappa(u-t)}\zeta^1_t$ of Ornstein-Uhlenbeck processes and an elementary integration for the last equality. Plugging this back into \eqref{eq:ex:trading rate} yields
	\begin{equation*}
	\dot{\varphi}^{\Lambda,1}_t=-\left(\sqrt{\Delta} - \tfrac{\delta}{2}I_d\right)\left(\varphi^{\Lambda,1}_t + (\sqrt{\Delta}+\tfrac{\delta}{2}I_d) (\sqrt{\Delta}+\tfrac{\delta}{2}I_d+\kappa)^{-1} \zeta^1_t\right).
	\end{equation*}
	Inserting this \eqref{eq:eqmu} in turn leads to the asserted Ornstein-Uhlenbeck dynamics:
	\begin{align*}
	d \mu^\Lambda_t &= \tfrac{(\gamma^1-\gamma^2)\Sigma}{2} (\dot \varphi^{\Lambda,1}_t dt+d\zeta^1_t) \\
	&= \tfrac{(\gamma^1-\gamma^2)\Sigma}{2} \left(\left(-\left(\sqrt{\Delta} - \tfrac{\delta}{2}I_d\right)\Big(\varphi^{\Lambda,1}_t + (\sqrt{\Delta}+\tfrac{\delta}{2}I_d) (\sqrt{\Delta}+\tfrac{\delta}{2}I_d+\kappa)^{-1} \zeta^1_t\Big)-\kappa \zeta^1_t\right) dt+d N_t\right)\\
	&=\Sigma \left(\sqrt{\Delta} - \tfrac{\delta}{2}I_d\right) \Sigma^{-1} \left(\tfrac{(\gamma^1-\gamma^2) \Sigma}{2} \Big(\kappa (\sqrt{\Delta}+\tfrac{\delta}{2} I_d +\kappa)^{-1}-(\sqrt{\Delta}-\tfrac{\delta}{2} I_d)^{-1}\kappa\Big)\zeta^1_t-\mu^\Lambda_t\right)dt \\
	&\quad +\tfrac{(\gamma^1-\gamma^2)\Sigma}{2}dN_t. \qedhere
		\end{align*}
	\end{proof}
	
For a single risky asset ($d=1$), the mean-reversion level of the liquidity premium in Corollary~\ref{cor:OU 2} can be rewritten as 
\begin{align*}
 \frac{(\gamma^2-\gamma^1) \Sigma \kappa(\delta+\kappa)}{2(\sqrt{\Delta}+\tfrac{\delta}{2} I_d +\kappa)(\sqrt{\Delta}-\tfrac{\delta}{2} I_d)}\zeta^1_t =2 \frac{\gamma_2 - \gamma_1}{\gamma^1+ \gamma_2} \Lambda \kappa (\kappa + \delta) \left(1 - \frac{\kappa}{\sqrt{\Delta}+\tfrac{\delta}{2} I_d +\kappa} \right) \zeta^1_t.
\end{align*}
Since $\gamma_1 < \gamma_2$, the coefficient of $\zeta^1_t$ in this expression is positive, so that the sign of the liquidity premium depends on Agent 1's risky exposure $\zeta^1_t$. If $\zeta^1_t>0$, mean-reversion implies that Agent~1's exposure will tend to decrease, so that this agent will want to buy back part of her negative hedging position in the risky asset. Conversely, Agent 2 will tend to sell risky shares. Since Agent 2 is more risk averse, her selling motive dominates and needs to be offset by an additional positive expected return to clear the market, in line with the sign of the above expression.  

\section{Conclusion}
\label{sec:conclusion}
In this paper, we develop a tractable risk-sharing model that allows to study how trading costs are reflected in expected returns. In a continuous-time model populated by heterogenous mean-variance investors, we characterize the unique equilibrium by a system of coupled but linear FBSDEs. This system can be solved in terms of matrix power series, and leads to fully explicit equilibrium dynamics in a number of concrete settings.

If all agents are homogenous, positive liquidity premia are obtained if the asset supply expands over time. For a fixed asset supply but heterogenous agents, the sign of the liquidity premia compared to the frictionless case is determined by the trading needs of the more risk averse agents. Since these have a stronger motive to trade, they need to compensate their more risk-tolerant counterparties accordingly. The sluggishness of illiquid portfolios also introduces autocorrelation into the corresponding equilibrium expected returns even for fundamentals with independent increments.

Several extensions of the present model are intriguing directions for further research. One important direction concerns more general specifications of preferences (e.g., exponential rather than quadratic utilities) or trading costs (e.g., proportional instead of quadratic). Such variations are bound to destroy the linearity of the corresponding optimality conditions, but might still lead to tractable results in the small-cost limit similarly as for models with exogenous prices \cite{soner.touzi.13,moreau.al.15,kallsen.muhlekarbe.17}. Another important direction for further research concerns extensions where the price volatility is no longer assumed to be exogenously given, but is instead determined as an output of the equilibrium. However, even the simplest versions of such models are also bound to lead to nonlinear FBSDEs.

\appendix

\section{Existence and Uniqueness of Linear FBSDEs}\label{a:FBSDE}

For the determination of both individually optimal trading strategies in Section~\ref{s:indopt} and equilibrium returns in Section~\ref{s:eqopt}, this appendix develops existence and uniqueness results for systems of coupled but linear FBSDEs:\footnote{Due to the degeneracy of the forward component \eqref{eq:fbsde1}, general FBSDE theory as in \cite{delarue.02} only yields local existence. However, the direct arguments developed below allow to establish global existence and also lead to explicit representations of the solution in terms of matrix power series.}
\begin{align}
d\varphi_t &=\dot{\varphi}_t dt, \quad \varphi_0=0, \quad t \in \scr{T}, \label{eq:fbsde1}\\
d\dot{\varphi}_t&=dM_t +B\left(\varphi_t-\xi_t\right)dt + \delta \dot \varphi_t dt, \quad t \in \scr{T}, \label{eq:fbsde2}
\end{align}
where $B \in \mathbb{R}^{\ell \times \ell}$ has only positive eigenvalues, $\delta \geq 0$, and $\xi \in \scr{L}^2_\delta$. If $T<\infty$, \eqref{eq:fbsde2} is complemented by the \emph{terminal condition} 
\begin{equation}
\label{eq:fbsde3:terminal}
\dot \varphi_T = 0.
\end{equation}
If $T = \infty$, we assume that $\delta > 0$ and the terminal condition is replaced by the \emph{transversality conditions} implicit in $\varphi, \dot \varphi \in \scr{L}^2_\delta$ for $\delta>0$. 
 A \emph{solution} of (\ref{eq:fbsde1}--\ref{eq:fbsde2}) is a triple $(\varphi, \dot \varphi, M)$ for which $\varphi, \dot \varphi \in  \scr{L}^2_\delta$ and $M$ is a martingale on $\scr{T}$ with finite second moments.

We first consider the infinite time-horizon case. In this case, the linear FBSDEs (\ref{eq:fbsde1}-\ref{eq:fbsde2}) can be solved using matrix exponentials similarly as in \cite{garleanu.pedersen.16}. To this end, we first establish a technical result stating that the martingale $M$ appearing in the solution of the FBSDE (\ref{eq:fbsde1}--\ref{eq:fbsde2}) automatically belongs to $\scr{M}^2_\delta$:

\begin{proposition}
\label{prop:BSDE:martingale}
Let $T = \infty$. If $(\varphi, \dot  \varphi, M)$ is a solution to the FBSDE (\ref{eq:fbsde1}--\ref{eq:fbsde2}), then $M \in \scr{M}^2_\delta$.
\end{proposition}

\begin{proof}
Let $(\varphi, \dot  \varphi, M)$ be a solution to the FBSDE (\ref{eq:fbsde1}--\ref{eq:fbsde2}), where $\varphi, \dot \varphi \in  \scr{L}^2_\delta$ and $M$ is a martingale on $[0, \infty)$ with finite second moments.
Fix $t \in (0, \infty)$. Then integration by parts yields
\begin{equation}
\label{eq:pf:prop:BSDE:martingale:int parts}
e^{-\delta t} \dot \varphi_t = \dot \varphi_0 + \int_0^t e^{-\delta s} dM_s + \int_0^t e^{-\delta s} B (\varphi_s - \xi_s) ds.
\end{equation}
Let $\Vert \cdot \Vert_2$ be the Euclidean norm in $\mathbb{R}^{\ell}$ and $\Vert \cdot \Vert_{\max}$ the maximum norm on $\mathbb{R}^{\ell \times \ell}$. Rearranging \eqref{eq:pf:prop:BSDE:martingale:int parts} and using subsequently the elementary inequality $(a + b + c)^2 \leq 3(a^2 + b^2 + c^2)$ for $a, b, c \in \mathbb{R}$, the elementary estimate $\Vert A x \Vert_2 \leq \sqrt{\ell} \Vert A \Vert_{\max} \Vert x \Vert_2$, and Jensen's inequality gives
\begin{align}
\left\Vert \int_0^t  e^{-\delta s} dM_s \right\Vert_2^2 &\leq 3 \left(\Vert \dot \varphi_0 \Vert_2^2 + e^{- 2 \delta t} \Vert \dot \varphi_t \Vert^2_2 + \left(\int_0^t e^{-\delta s} \sqrt{\ell} \Vert B \Vert_{\max} \Vert \varphi_s-\xi_s \Vert_2 ds\right)^2 \right) \notag \\
&\leq 3 \left(\Vert \dot \varphi_0 \Vert_2^2 + e^{- \delta t} \Vert \dot \varphi_t \Vert^2_2 + \frac{\ell}{\delta} \Vert B \Vert^2_{\max} \int_0^t e^{-\delta s}  \Vert \varphi_s-\xi_s \Vert^2_2 ds \right). 
\label{eq:pf:prop:BSDE:martingale:estimate 1}
\end{align}
The definitions of the maximum and Euclidean norms, the estimate $|[N^1, N^2 ]| \leq \tfrac{1}{2} ([N^1 ]+ [N^2])$ for real-valued local martingales $N^1$ and $N^2$, and It\^o's isometry give, for $t \in (0, \infty)$,
\begin{align}
E\left[\Big\Vert \int_0^t e^{-2 \delta s } d[M]_s \Big \Vert_{\max} \right] &= E\left[ \max_{i, j \in \{1, \ldots, \ell\}} \int_0^t e^{-2 \delta s } d[M^i, M^j]_s \right] \notag \\
&\leq \frac{1}{2} E\left[ \max_{i \in \{1, \ldots, \ell\}} \int_0^t e^{-2 \delta s } d[M^i]_s + \max_{j \in \{1, \ldots, \ell\}} \int_0^t e^{-2 \delta s } d[M^j]_s \right] \notag \\
&= E\left[ \max_{i \in \{1, \ldots, \ell\}} \int_0^t e^{-2 \delta s } d[M^i]_s \right] \leq E\left[ \sum_{i =1}^\ell \int_0^t e^{-2 \delta s } d[M^i]_s \right] \notag \\
&= E\left[ \sum_{i =1}^\ell \left(\int_0^t e^{-\delta s } dM^i_s\right)^2 \right] = E\left[\left\Vert \int_0^t  e^{-\delta s} dM_s \right\Vert_2^2\right].
\label{eq:pf:prop:BSDE:martingale:estimate 2}
\end{align}
Since $\dot \varphi \in \scr{L}^2_\delta$, there exists an increasing sequence $(t_k)_{k \in \mathbb{N}}$ with $\lim_{k \to \infty} t_k =\infty$ such that 
\begin{equation}
\label{eq:pf:prop:BSDE:martingale:limit}
\lim_{n \to \infty} E\left[e^{- \delta t_k} \Vert \dot \varphi_{t_k} \Vert^2_2 \right] = 0.
\end{equation}
Monotone convergence, (\ref{eq:pf:prop:BSDE:martingale:estimate 1}-\ref{eq:pf:prop:BSDE:martingale:limit}) and $\varphi, \xi \in \scr{L}^2_\delta$ in turn yield
\begin{align*}
E\left[\Big\Vert \int_0^\infty e^{-2 \delta s } d[M]_s \Big \Vert_{\max} \right] &= \lim_{k \to \infty} E\left[\Big\Vert \int_0^{t_k} e^{-2 \delta s } d[M]_s \Big\Vert_{\max} \right] \leq  \lim_{k \to \infty} E\left[\left\Vert \int_0^{t_k}  e^{-\delta s} dM_s \right\Vert_2^2\right] \\
&\leq 3 \ell \left(\Vert \dot \varphi_0 \Vert_2^2 + \frac{\ell}{\delta} \Vert B \Vert^2_{\max} E\left[ \int_0^\infty e^{-\delta s}  \Vert \varphi_s-\xi_s \Vert^2_2 ds \right] \right) < \infty.
\end{align*}
Thus, $M \in \scr{M}^2_\delta$ as claimed.
\end{proof}

\begin{theorem}\label{thm:fbsde:infinite}
	Suppose that $T = \infty$, $\delta >  0$, and the matrix $B$ from \eqref{eq:fbsde2} has only positive eigenvalues. Set $\Delta = B + \frac{\delta^2}{4} I_\ell$. Then, the unique solution of the FBSDE (\ref{eq:fbsde1}--\ref{eq:fbsde2}) is given by 
	\begin{equation}\label{eq:varphi:infinite}
	\varphi_t= \int_0^t \left(e^{-(\sqrt{\Delta} - \frac{\delta}{2} I_\ell)(t-s)} \bar{\xi}_s \right)ds,
	\end{equation}
	where
	\begin{align}
	\label{def:bar xi:infinite}
	\bar{\xi}_t &=\left(\sqrt{\Delta} - \tfrac{\delta}{2} I_\ell\right) E\left[\int_t^\infty \left(\sqrt{\Delta} + \tfrac{\delta}{2} I_\ell\right) e^{-(\sqrt{\Delta}+ \frac{\delta}{2} I_\ell) (s-t)}\xi_s ds\,\Big|\, \mathcal{F}_t\right].
	\end{align}
\end{theorem}

\begin{proof}
Let $(\varphi,\dot \varphi, M)$ be a solution to the FBSDE (\ref{eq:fbsde1}--\ref{eq:fbsde2}) and define
	\begin{equation*}
	\tilde \varphi_t := e^{-\frac{\delta}{2}t} \varphi_t.
	\end{equation*}
	Using that
	\begin{align}\label{eq:changevar:infinite}
	\dot {\tilde\varphi}_t &= -\tfrac{\delta}{2} \tilde \varphi_t + e^{-\frac{\delta}{2} t} \dot \varphi_t,  \qquad d \dot {\tilde\varphi}_t = -\tfrac{\delta}{2} \dot {\tilde \varphi}_t dt + e^{-\frac{\delta}{2} t} d \dot  \varphi_t -\tfrac{\delta}{2} e^{-\frac{\delta}{2} t} \dot \varphi_t dt
	\end{align}
	and the FBSDE (\ref{eq:fbsde1}-\ref{eq:fbsde2}) for $(\varphi,\dot{\varphi})$, it follows that $(\tilde \varphi,\dot{\tilde \varphi})$ solves the FBSDE
	\begin{align}
	d\tilde \varphi_t &=\dot{\tilde \varphi}_t dt, \quad \tilde \varphi_0=0, \quad t \in \scr{T}, \label{eq:fbsde tilde 1:infinite}\\
	d\dot{\tilde \varphi}_t&=d\tilde M_t +B\left(\tilde \varphi_t- \tilde \xi_t\right)dt + \tfrac{\delta^2}{4} \tilde \varphi_t dt, \quad t \in \scr{T}, \label{eq:fbsde tilde 2}
	\end{align}
	where $d\tilde M_t = e^{-\frac{\delta}{2} t} dM_t$ and $\tilde \xi_t = e^{-\frac{\delta}{2} t} \xi_t$. In matrix notation, this equation can be rewritten as
	$$
	d(\tilde \varphi_t,\dot{\tilde \varphi}_t)^\top=C_1d\tilde M_t+C_2(\tilde \varphi_t,\dot{\tilde \varphi}_t)^\top dt-C_3\tilde \xi_tdt,
	$$
	with
	$$
	C_1=\begin{pmatrix} 0 \\ I_\ell \end{pmatrix}, \quad C_2=\begin{pmatrix} 0 & I_\ell \\ \Delta & 0 \end{pmatrix}, \quad C_3=\begin{pmatrix} 0 \\ B \end{pmatrix}.
	$$
	Integration by parts shows
	$$d(e^{-C_2 t}(\tilde \varphi_t,\dot{\tilde \varphi}_t)^\top)=e^{-C_2 t}C_1d\tilde M_t-e^{-C_2 t}C_3\tilde \xi_t dt,$$
	and in turn
	\begin{equation}\label{eq:inter:infinite:A}
	e^{-C_2 u}\begin{pmatrix} \tilde \varphi_u \\ \dot{\tilde \varphi}_u \end{pmatrix}= e^{-C_2 t}\begin{pmatrix} \tilde \varphi_t \\ \dot{\tilde \varphi}_t \end{pmatrix}+\int_t^u e^{-C_2 s}C_1d \tilde M_s -\int_t^u e^{-C_2 s}C_3\tilde \xi_s ds, \quad \mbox{for $t< u < \infty$.}
	\end{equation}
Multiplying \eqref{eq:inter:infinite:A} by the matrix
\begin{equation*}
\tilde C_2=\begin{pmatrix}  I_\ell & \sqrt{\Delta}^{-1} \\ \sqrt{\Delta} & I_\ell \end{pmatrix}
\end{equation*}
and setting  $H(t)=\tilde C_2 e^{-C_2 t}$ yields
	\begin{equation}\label{eq:inter:infinite}
H(u) \begin{pmatrix} \tilde \varphi_u \\ \dot{\tilde \varphi}_u \end{pmatrix}= H(t)\begin{pmatrix} \tilde \varphi_t \\ \dot{\tilde \varphi}_t \end{pmatrix}+\int_t^u H(s)C_1d \tilde M_s -\int_t^u \tilde H(s) C_3\tilde \xi_s ds, \quad \mbox{for $t< u < \infty$.}
\end{equation}	
It follows by induction that
\begin{equation*}
\tilde C_2 (- C_2 t)^{2n} =\begin{pmatrix}  \Delta^n & \Delta^{n - \tfrac{1}{2}}  \\ \Delta^{n + \tfrac{1}{2}} & \Delta^n \end{pmatrix} t^{2n},  \quad \tilde C_2 (- C_2 t)^{2n+1} =-\begin{pmatrix}  \Delta^{n+\frac{1}{2}} & \Delta^{n }  \\ \Delta^{n+1} & \Delta^{n+\frac{1}{2}} \end{pmatrix} t^{2n +1}, \quad n \geq 0.
\end{equation*}
Now, the power series for the exponential function allows to deduce
$$
	H(t)=\begin{pmatrix} e^{- \sqrt{\Delta} t} &  \sqrt{\Delta}^{-1} e^{- \sqrt{\Delta} t}\\ \sqrt{\Delta} e^{- \sqrt{\Delta} t} & e^{- \sqrt{\Delta} t} \end{pmatrix}.
	$$
Together with \eqref{eq:inter:infinite}, it follows that
\begin{align}
	\sqrt{\Delta} e^{- \sqrt{\Delta} u}\tilde \varphi_u + e^{- \sqrt{\Delta} u} \dot{\tilde \varphi}_u &= \sqrt{\Delta}e^{- \sqrt{\Delta} t} \tilde \varphi_t + e^{- \sqrt{\Delta} t} \dot{\tilde \varphi}_t \notag \\
	&\quad +\int_t^u e^{- \sqrt{\Delta} s} d\tilde M_s -\int_t^u e^{- \sqrt{\Delta} s} B \tilde \xi_s ds, \quad \mbox{for $t< u < \infty$.}
	\label{eq:pf:fbsde:infinite:prelimit}
\end{align}
By the assumption that $\varphi, \dot \varphi \in \scr{L}^1_\delta \subset  \scr{L}^2_\delta$, \eqref{eq:changevar:infinite} and the fact that all eigenvalues of $\sqrt{\Delta}$ are greater or equal than $\delta/2$ (because $B$ has only nonnegative eigenvalues), there exists an increasing sequence $(u_k)_{k \in \mathbb{N}}$ with $\lim_{k \to \infty} u_k = + \infty$ along which the left-hand side of  \eqref{eq:pf:fbsde:infinite:prelimit} converges a.s.~to zero. Moreover, since $M \in  \scr{M}^2_\delta$ by Proposition \ref{prop:BSDE:martingale} and $\tilde \xi \in  \scr{L}^2_\delta$, the martingale convergence theorem and monotone convergence (together with Jensen's inequality) -- also using that all eigenvalues of $\sqrt{\Delta}$ are greater or equal than $\delta/2$ -- imply that for $u \to \infty$ (and a fortiori along $(u_k)_{k \in \mathbb{N}}$) the right hand side of \eqref{eq:pf:fbsde:infinite:prelimit} converges a.s. to
\begin{equation}
	\label{eq:pf:fbsde:infinite:limit}
\sqrt{\Delta}e^{- \sqrt{\Delta} t} \tilde \varphi_t + e^{- \sqrt{\Delta} t} \dot{\tilde \varphi}_t +\int_t^\infty e^{- \sqrt{\Delta} s} d\tilde M_s -\int_t^\infty e^{- \sqrt{\Delta} s} B \tilde \xi_s ds.
\end{equation}
Together, these two limits show that \eqref{eq:pf:fbsde:infinite:limit} vanishes. Multiplying \eqref{eq:pf:fbsde:infinite:limit}  by $e^{\sqrt{\Delta} t}$, rearranging, and taking conditional expectations (using again that all eigenvalues of $\sqrt{\Delta}$ are larger than or equal to $\delta/2$) we obtain
	\begin{align*}
\dot {\tilde \varphi}_t  &= E\left[\int_t^\infty e^{-\sqrt{\Delta} (s-t)} B e^{-\frac{\delta}{2} s} \xi_s ds\Big| \mathcal{F}_t\right] - \sqrt{\Delta} \tilde \varphi_t.
	\end{align*}
	Now, \eqref{eq:changevar:infinite} and rearranging give
	\begin{align*}
\dot \varphi_t = E\left[\int_t^\infty e^{-\sqrt{\Delta} (s-t)} B e^{-\frac{\delta}{2} (s-t)} \xi_s ds\Big| \mathcal{F}_t\right] - \left(\sqrt{\Delta} - \tfrac{\delta}{2} I_\ell\right) \varphi_t.
	\end{align*}
Finally, since $B$ commutes with $e^{-\sqrt{\Delta} (s-t)}$ (as $B = (\sqrt{\Delta} - \frac{\delta}{2} I_\ell)(\sqrt{\Delta} + \frac{\delta}{2} I_\ell) $ and by Lemma \ref{lem:matrix function:properties}(a)) it follows that
	\begin{equation}\label{eq:ODE:infinite}
	\dot {\varphi}_t = \bar{\xi}_t- \left(\sqrt{\Delta} - \tfrac{\delta}{2} I_\ell\right) \varphi_t.
	\end{equation}
	By the variations of constants formula, this linear (random) ODE has the unique solution \eqref{eq:varphi:infinite}. If a solution of the FBSDE (\ref{eq:fbsde1}-\ref{eq:fbsde2}) exists, it therefore must be of the form \eqref{eq:varphi:infinite}.
	
	It remains to verify that \eqref{eq:varphi:infinite} indeed solves the FBSDE (\ref{eq:fbsde1}-\ref{eq:fbsde2}).  To this end, we first show that $\bar \xi \in \scr{L}^2_\delta$. Indeed, denote by $\Vert \cdot \Vert_2$ both the Euclidean norm in $\mathbb{R}^{\ell}$ and the spectral norm in $\mathbb{R}^{\ell \times \ell}$. Since all eigenvalues of $B$ are positive, there is $\varepsilon > 0$ such that all eigenvalues of $\sqrt{\Delta} + \tfrac{\delta}{2} I_\ell$ are greater or equal that $\delta + \varepsilon$. Hence, by Lemma \ref{lem:matrix function:properties}(c) and the definition of the spectral norm, it follows that
	\begin{equation*}
	\left \Vert e^{-(\sqrt{\Delta} + \tfrac{\delta}{2} I_\ell)t} \right \Vert_2 \leq e^{-(\delta + \varepsilon) t}, \quad t \in [0, \infty).
	\end{equation*}
	Thus, by the definition of $\bar \xi$ in \eqref{def:bar xi:infinite}, the fact that $B = (\sqrt{\Delta} - \frac{\delta}{2} I_\ell)(\sqrt{\Delta} + \frac{\delta}{2} I_\ell) $, Jensen's inequality and Fubini's theorem, we obtain
	\begin{align*}
E\left[\int_0^\infty e^{-\delta t} \Vert \bar \xi_t \Vert^2_2 d t\right] &\leq \frac{\Vert B \Vert^2_2}{\delta + \epsilon} \int_0^\infty e^{-\delta t} \int_t^\infty e^{-(\delta + \varepsilon) (s-t)} E\left[\Vert\xi_s \Vert^2_2\right] ds dt \\
&\leq \frac{\Vert B \Vert^2_2}{\delta + \varepsilon} \int_0^\infty  \left(\int_0^s  e^{ \varepsilon t} dt \right) e^{-(\delta + \varepsilon) s} E\left[\Vert\xi_s \Vert^2_2\right] ds \notag \\
&\leq \frac{\Vert B \Vert^2_2}{\epsilon(\delta + \varepsilon)} \int_0^\infty  e^{-\delta s} E\left[\Vert\xi_s \Vert^2_2\right] ds < \infty.
	\end{align*}
	Next, we show that $\varphi \in \scr{L}^2_\delta$. Arguing similarly as above, we have 
		\begin{equation*}
	\left \Vert e^{-(\sqrt{\Delta} - \tfrac{\delta}{2} I_\ell)t} \right \Vert_2 \leq e^{-\varepsilon t}, \quad t \in [0, \infty).
	\end{equation*}
	Thus, by the definition of $\varphi$ in \eqref{eq:varphi:infinite}, Jensen's inequality and Fubini's theorem and since $\bar \xi \in \scr{L}^2_\delta$ by the above arguments, we obtain
		\begin{align*}
	E\left[\int_0^\infty e^{-\delta t} \Vert \varphi_t \Vert^2_2 d t\right] &\leq \frac{1}{\varepsilon} \int_0^\infty e^{-\delta t} \int_0^t e^{-\epsilon (t -s)} E\left[\Vert \bar \xi_s \Vert^2_2\right] ds dt \\
	&\leq \frac{1}{\epsilon} \int_0^\infty  \left(\int_s^\infty  e^{ -(\delta + \varepsilon) t} dt \right) e^{ \varepsilon s} E\left[\Vert \bar \xi_s \Vert^2_2\right] ds \notag \\
	&= \frac{1}{\epsilon(\epsilon \delta + \epsilon)} \int_0^\infty  e^{-\delta s} E\left[\Vert \bar \xi_s \Vert^2_2\right] ds < \infty.
	\end{align*}
By definition, we have $\varphi_0=0$. Next, integration by parts shows that $\dot{\varphi}$ satisfies the ODE
	\eqref{eq:ODE:infinite}, and this yields $\dot{\varphi} \in \scr{L}^2_\delta$ (because $\varphi, \bar \xi \in \scr{L}^2_\delta$). Define the $\mathbb{R}^\ell$-valued square-integrable martingale $(\bar{M}_t)_{t \in [0, \infty)}$ by\footnote{Note that $\int_0^\infty e^{-\sqrt{\Delta} s} B \tilde \xi_s ds$ is square integrable because $\xi \in \scr{L}^2_\delta$ and all eigenvalues of $\sqrt{\Delta}$ are greater than $\delta/2$.}
	\begin{equation*}
	\bar{M}_t= E\left[\int_0^\infty e^{-\sqrt{\Delta} s} B \tilde \xi_s ds\Big| \mathcal{F}_t\right],
	\end{equation*}
where $\tilde \varphi_t := e^{-\frac{\delta}{2}t} \varphi_t$ and $\tilde \xi_t := e^{-\frac{\delta}{2}t} \xi_t$ as before. Then multiplying \eqref{eq:ODE:infinite} by the matrix $e^{-(\sqrt{\Delta} + \frac{\delta}{2}I_\ell)t}$ and using \eqref{eq:changevar:infinite} as well as $\Delta - \frac{\delta^2}{4} I_\ell = B$ gives, after some rearrangement,
	\begin{align*}
	e^{-\sqrt{\Delta} t} \dot {\tilde \varphi}_t  &= \bar M_t - \int_0^t e^{-\sqrt{\Delta} s} B \tilde \xi_s ds - \sqrt{\Delta} e^{-\sqrt{\Delta} t} \tilde \varphi_t.
	\end{align*}
	Taking differentials, we therefore obtain
$$
	-\sqrt{\Delta} e^{-\sqrt{\Delta} t} \dot {\tilde \varphi}_t dt + e^{-\sqrt{\Delta} t} d \dot{\tilde \varphi}_t = d \bar M_t - e^{-\sqrt{\Delta} t} B \tilde \xi_t dt -\sqrt{\Delta} e^{-\sqrt{\Delta} t}  \dot {\tilde \varphi}_t dt + \Delta e^{-\sqrt{\Delta} t}  \tilde \varphi_t dt.
$$
	Rearranging, multiplying by $e^{\sqrt{\Delta} t}$ and using that $\sqrt{\Delta}$ and $e^{\sqrt{\Delta} t}$ commute, it follows that 
	\begin{align*}
 d\dot{\tilde \varphi}_t  = e^{\sqrt{\Delta} t} d \bar M_t - B \tilde \xi_t dt + \Delta \tilde \varphi_t dt.
\end{align*}
	Finally, again taking into account~\eqref{eq:changevar:infinite} and defining the martingale $M$ (which has finite second moments) by
	\begin{equation*}
	d M_t = e^{(\sqrt{\Delta} + \frac{\delta}{2} I_\ell) t} d \bar M_t, \quad M_0 = \bar M_0,
	\end{equation*}
	we obtain that $\varphi$ from  \eqref{eq:varphi:infinite} indeed satisfies (\ref{eq:fbsde1}--\ref{eq:fbsde2}).
\end{proof}

Let us briefly sketch the financial interpretation of the solution; cf.~\cite{garleanu.pedersen.16} for more details. In the context of individually optimal trading strategies (cf.~Lemma~\ref{thm:indopt1}), the ODE \eqref{eq:ODE:infinite} describes the optimal trading rate. It prescribes to trade with a constant \emph{relative} speed $\sqrt{\Delta} - \tfrac{\delta}{2} I_\ell$ towards the \emph{target portfolio}
$$
	(\sqrt{\Delta} - \tfrac{\delta}{2} I_\ell)^{-1} \bar{\xi}_t = E\left[\int_t^\infty \left(\sqrt{\Delta} + \tfrac{\delta}{2} I_\ell\right) e^{-(\sqrt{\Delta}+ \frac{\delta}{2} I_\ell) (s-t)}\xi_s ds\,\Big|\, \mathcal{F}_t\right].
$$	
In the context of Lemma~\ref{thm:indopt1}, this is an average of the future values of the frictionless optimal trading strategy $\xi$, computed using an exponential discounting kernel. As the trading costs tend to zero, the discount rate tends to infinity, and the target portfolio approaches the current value of the frictionless optimizer, in line with the small-cost asymptotics of \cite{moreau.al.15}.

\medskip

We now turn to the finite-horizon case. In order to satisfy the terminal condition $\dot{\varphi}_T=0$, the exponentials from Theorem~\ref{thm:fbsde:infinite} need to be replaced by appropriate hyperbolic functions in the one-dimensional case~\cite{bank.al.17}. In the present multivariate context, this remains true if these hyperbolic functions are used to define the corresponding ``primary matrix functions'' in the sense of Definition~\ref{def:matrix function}. The first step to make this precise is the following auxiliary result, which is applied for $\Delta =B+\frac{\delta^2}{4}I_\ell$ in Theorem~\ref{thm:fbsde} below:

\begin{lemma}\label{lem:G}
Let $\Delta \in \mathbb{R}^{\ell \times \ell}$. The matrix-valued function 
 \begin{equation}
G(t) = \sum_{n = 0}^\infty \frac{1}{(2n)!} \Delta^{n} (T - t)^{2n} \label{eq:G}
\end{equation}
is twice differentiable on $\mathbb{R}$ with derivative
$$\dot{G}(t)=-\sum_{n = 0}^\infty \frac{1}{(2n+1)!} \Delta^{n+1} (T - t)^{2n+1},$$
and solves the following ODE:
\begin{equation}\label{eq:ODEG}
\ddot{G}(t)=\Delta G(t),  \quad \mbox{with } G(T)=I_d \mbox{ and } \dot{G}(T)=0.
\end{equation}
Moreover, if the matrix $\Delta$ has only  positive eigenvalues then, in the sense of Definition~\ref{def:matrix function}, 
$$G(t) = \cosh(\sqrt{\Delta} (T-t)), \qquad \dot G(t) = -\sqrt{\Delta}\sinh(\sqrt{\Delta} (T-t));$$
for $\delta \geq 0$, the matrix $\Delta G(t) - \frac{\delta}{2} \dot{G}(t)$ is invertible for any $t \in [0, T]$ and, for any matrix norm $\Vert \cdot \Vert$,
\begin{equation}
\label{eq:lem:G:uniform}
\sup_{t \in [0, T]} \left\Vert \Big(\Delta G(t) - \frac{\delta}{2} \dot{G}(t)\Big)^{-1} \right\Vert < \infty. \qedhere
\end{equation}
\end{lemma}

\begin{proof}
Note that $\sum_{n = 0}^\infty \frac{1}{(2n)!} \Vert \Delta \Vert^n (T - t)^{2n} < \infty$ for any matrix norm $\Vert \cdot \Vert$. Whence, $G(t)$ is well defined for each $t \in \mathbb{R}$. By twice differentiating term by term, and estimating the resulting power series in the same way, it is readily verified that $G$ is twice continuously differentiable on $\mathbb{R}$, has the stated derivative, and is a solution of \eqref{eq:ODEG}.

Suppose now that $\Delta$ has only positive eigenvalues and $\delta \geq 0$. Then the first two additional claims follow from Definition~\ref{def:matrix function} via the fact that $B = (\sqrt{B})^2$ and the series representation of the smooth functions $\cosh$ and $\sinh$. The final claim follows from Lemma~\ref{lem:matrix function:properties}(c) and (d) since, for fixed $x \in (0,\infty)$,
\begin{equation*}
\inf_{t \in [0, T]} x \cosh(x (T -t)) + \tfrac{\delta}{2} x \sinh(x (T -t)) \geq x > 0.
\end{equation*}
\end{proof}

The unique solution of our FBSDE can now be characterized using the function $G(t)$ from Lemma \ref{lem:G} as follows:

\begin{theorem}\label{thm:fbsde}
Suppose that $T < \infty$ and that the matrix $\Delta=B + \frac{\delta^2}{4} I_\ell$ has only positive eigenvalues. Then, the unique solution of the FBSDE (\ref{eq:fbsde1}-\ref{eq:fbsde2}) with terminal condition \eqref{eq:fbsde3:terminal} is given by 
\begin{equation}\label{eq:varphi}
\varphi_t= \int_0^t \left(e^{-\int_s^t F(u)du} \bar{\xi}_s \right)ds,
\end{equation}
where\footnote{Note that the inverses are well defined by Lemma \ref{lem:G}.}
\begin{align}
F(t) &=-\left(\Delta G(t) -\frac{\delta}{2} \dot G(t)\right)^{-1}\ B \dot G(t)\label{eq:Fdef},\\
\bar{\xi}_t &=\left(\Delta G(t) -\frac{\delta}{2} \dot G(t)\right)^{-1}E\left[\int_t^T \left(\Delta G(s) - \frac{\delta}{2} \dot G(s) \right) B e^{-\frac{\delta}{2} (s-t)}\xi_s ds\Big| \mathcal{F}_t\right].
\end{align}
\end{theorem}

\begin{proof}
Let $(\varphi,\dot{\varphi})$ be a solution of the FBSDE (\ref{eq:fbsde1}--\ref{eq:fbsde2}) with terminal condition \eqref{eq:fbsde3:terminal} and set
\begin{equation*}
\tilde \varphi_t := e^{-\frac{\delta}{2}t} \varphi_t.
\end{equation*}
Using that
\begin{align}\label{eq:changevar}
\dot {\tilde\varphi}_t &= -\tfrac{\delta}{2} \tilde \varphi_t + e^{-\frac{\delta}{2} t} \dot \varphi_t,  \quad d \dot {\tilde\varphi}_t = -\tfrac{\delta}{2} \dot {\tilde \varphi}_t dt + e^{-\frac{\delta}{2} t} d \dot  \varphi_t -\tfrac{\delta}{2} e^{-\frac{\delta}{2} t} \dot \varphi_t dt
\end{align}
and the FBSDE (\ref{eq:fbsde1}--\ref{eq:fbsde2}) for $(\varphi,\dot{\varphi})$ with \eqref{eq:fbsde3:terminal}, it follows that $(\tilde \varphi,\dot{\tilde \varphi})$ solves the FBSDE
\begin{align}
d\tilde \varphi_t &=\dot{\tilde \varphi}_t dt, \quad \tilde \varphi_0=0, \quad t \in [0, T]\label{eq:fbsde tilde 1:finite}\\
d\dot{\tilde \varphi}_t&=d\tilde M_t +B\left(\tilde \varphi_t- \tilde \xi_t\right)dt + \tfrac{\delta^2}{4} \tilde \varphi_t dt, \quad t \in [0, T],
\end{align}
with terminal condition
\begin{equation}
\label{eq:pf:fbsde:terminal}
\dot{\tilde \varphi}_T=-\tfrac{\delta}{2} \tilde \varphi_T.
\end{equation}
Here $d\tilde M_t =e^{-\frac{\delta}{2} t} dM_t$ is a square-integrable martingale (because $M$ is) and $\tilde \xi_t = e^{-\frac{\delta}{2} t} \xi_t$. In matrix notation, this equation can be rewritten as
$$
d(\tilde \varphi_t,\dot{\tilde \varphi}_t)^\top=C_1d\tilde M_t+C_2(\tilde \varphi_t,\dot{\tilde \varphi}_t)^\top dt-C_3\tilde \xi_tdt,
$$
with
$$
C_1=\begin{pmatrix} 0 \\ I_\ell \end{pmatrix}, \quad C_2=\begin{pmatrix} 0 & I_\ell \\ \Delta & 0 \end{pmatrix}, \quad C_3=\begin{pmatrix} 0 \\ B \end{pmatrix}.
$$
Integration by parts shows
$$d(e^{C_2(T-t)}(\tilde \varphi_t,\dot{\tilde \varphi}_t)^\top)=e^{C_2(T-t)}C_1d\tilde M_t-e^{C_2(T-t)}C_3\tilde \xi_t dt,$$
and in turn
\begin{equation}\label{eq:inter}
\begin{pmatrix} \tilde \varphi_T \\ \dot{\tilde \varphi}_T \end{pmatrix}= e^{C_2(T-t)}\begin{pmatrix} \tilde \varphi_t \\ \dot{\tilde \varphi}_t \end{pmatrix}+\int_t^T e^{C_2(T-s)}C_1d \tilde M_s -\int_t^T e^{C_2(T-s)}C_3\tilde \xi_s ds.
\end{equation}
Set $H(t)=e^{C_2(T-t)}$ and note that $$H(t)=\begin{pmatrix} G(t) &  -\Delta^{-1}\dot G(t)\\ -\dot G(t) & G(t)\end{pmatrix},$$
for the function $G(t)$ from Lemma \ref{lem:G}, as is readily verified by induction. Together with \eqref{eq:inter}, it follows that
\begin{align*}
\tilde \varphi_T &= G(t) \tilde \varphi_t -\Delta^{-1} \dot G(t) \dot{\tilde \varphi}_t - \int_t^T \Delta^{-1} \dot G(s) d\tilde M_s + \int_t^T \Delta^{-1} \dot G(s)  B \tilde \xi_s ds, \\
\dot{\tilde \varphi}_T &= -\dot{G}(t) \tilde \varphi_t + G(t) \dot{\tilde \varphi}_t +\int_t^T G(s) d\tilde M_s -\int_t^T G(s) B \tilde \xi_s ds.
\end{align*}
Since $\dot{\tilde \varphi}_T=-\frac{\delta}{2} \tilde \varphi_T$ by \eqref{eq:pf:fbsde:terminal}, this in turn yields
\begin{align*}
0 = &\left(\tfrac{\delta}{2} G(t) - \dot G(t) \right) \tilde \varphi_t + \left(-\tfrac{\delta}{2} \Delta^{-1} \dot G(t) +  G(t) \right) \dot {\tilde \varphi}_t  +\int_t^T \left(-\tfrac{\delta}{2} \Delta^{-1} \dot G(s) + G(s) \right) d\tilde M_s \\
&+ \int_t^T \left(\tfrac{\delta}{2} \Delta^{-1} \dot G(s) - G(s)\right) B \tilde \xi_s ds.
\end{align*}
Multiplying this equation by $\Delta$ and taking conditional expectations gives
\begin{align*}
\left(\Delta G(t) -\tfrac{\delta}{2} \dot G(t)\right) \dot {\tilde \varphi}_t  &= E\left[\int_t^T \left(\Delta G(s) - \tfrac{\delta}{2} \dot G(s) \right) B \tilde \xi_s ds\Big| \mathcal{F}_t\right] +\left(\Delta \dot G(t) -\tfrac{\delta}{2} \Delta G(t)\right) \tilde \varphi_t.
\end{align*}
Now, using \eqref{eq:changevar} and rearranging, it follows that
\begin{align*}
\left(\Delta G(t) -\tfrac{\delta}{2} \dot G(t)\right) e^{-\frac{\delta}{2} t} \dot \varphi_t  &= E\left[\int_t^T \left(\Delta G(s) - \tfrac{\delta}{2} \dot G(s) \right) B e^{-\frac{\delta}{2} s} \xi_s ds\Big| \mathcal{F}_t\right] \\
&\quad+\left(\Delta - \tfrac{\delta^2}{4} I_\ell\right) \dot G(t) e^{-\frac{\delta}{2} t} \varphi_t.
\end{align*}
After multiplying with the inverse of $\left(\Delta G(t) -\frac{\delta}{2} \dot G(t)\right)$ (which exists by Lemma \ref{lem:G}) and using that $\Delta - \frac{\delta^2}{4} I_\ell = B$, this leads to
\begin{equation}\label{eq:ODE}
 \dot {\varphi}_t = \bar{\xi}_t-F(t)\varphi_t.
 \end{equation}
By the variations of constants formula, this linear (random) ODE has the unique solution \eqref{eq:varphi}. If a solution of the FBSDE (\ref{eq:fbsde1}-\ref{eq:fbsde2}) exists, it therefore must be of the form \eqref{eq:varphi}.

It remains to verify that \eqref{eq:varphi} indeed solves the FBSDE (\ref{eq:fbsde1}-\ref{eq:fbsde2}). First, note that $\bar \xi \in \scr{L}^2_\delta$ by the fact that $\xi \in \scr{L}^2_\delta$ and the estimate \eqref{eq:lem:G:uniform}, and in turn $\varphi \in \scr{L}^2_\delta$. Moreover, by definition, we have $\varphi_0=0$. Next, integration by parts shows that $\dot{\varphi}$ satisfies the ODE
\eqref{eq:ODE}, and this yields $\dot{\varphi} \in \scr{L}^2_\delta$ (because $\varphi, \bar \xi \in \scr{L}^2_\delta$) and $\dot{\varphi}_T=0$ (because $G(T)=I$ and $\dot{G}(T)=0$). Define the $\mathbb{R}^\ell$-valued square-integrable martingale $(\bar{M}_t)_{t \in [0, T]}$ by
\begin{equation*}
\bar{M}_t= E\left[\int_0^T \left(\Delta G(s) - \tfrac{\delta}{2} \dot G(s) \right) B \tilde \xi_s ds\,\Big|\, \mathcal{F}_t\right],
\end{equation*}
where $\tilde \varphi_t := e^{-\frac{\delta}{2}t} \varphi_t$ and $\tilde \xi_t := e^{-\frac{\delta}{2}t} \xi_t$ as before. Then, multiplying \eqref{eq:ODE} by $\left(\Delta G(t) -\tfrac{\delta}{2} \dot G(t)\right)$ and using \eqref{eq:changevar} as well as $\Delta - \frac{\delta^2}{4} I_\ell = B$ gives, after some rearrangement,
\begin{align*}
\left(\Delta G(t) -\tfrac{\delta}{2} \dot G(t)\right) \dot {\tilde \varphi}_t  &= \bar M_t - \int_0^t \left(\Delta G(s) - \tfrac{\delta}{2} \dot G(s) \right) B \tilde \xi_s ds +\left(\Delta \dot G(t) -\tfrac{\delta}{2} \Delta G(t)\right) \tilde \varphi_t.
\end{align*}
Taking differentials, we therefore obtain
\begin{align*}
&\left(\Delta \dot G(t) -\tfrac{\delta}{2} \ddot{G}(t)\right) \dot {\tilde \varphi}_t dt + \left(\Delta G(t) -\tfrac{\delta}{2} \dot{G}(t)\right) d \dot{\tilde \varphi}_t   \\
&\qquad= d \bar M_t - \left(\Delta G(t) - \tfrac{\delta}{2} \dot G(t) \right) B \tilde \xi_t dt +\left(\Delta \dot G(t) -\tfrac{\delta}{2} \Delta G(t)\right) \dot{\tilde \varphi}_t dt + \left(\Delta \ddot G(t) -\tfrac{\delta}{2} \Delta \dot G(t)\right) \tilde \varphi_t dt.
\end{align*}
Using that $\ddot{G}(t) = \Delta G(t)$ by the ODE \eqref{eq:ODEG} and taking into account that $\Delta$ commutes with both $G(t)$ and $\dot G(t)$ by Lemma \ref{lem:matrix function:properties}(a), it follows that
\begin{align*}
\left(\Delta G(t) -\tfrac{\delta}{2} \dot{G}(t)\right) d \dot{\tilde \varphi}_t  = d \bar M_t - \left(\Delta G(t) - \tfrac{\delta}{2} \dot G(t) \right) B \tilde \xi_t dt  + \left(\Delta G(t) -\tfrac{\delta}{2} G(t)\right) \Delta \tilde \varphi_t dt.
\end{align*}
Now, multiplying with the inverse of $\left(\Delta G(t) -\frac{\delta}{2} \dot G(t)\right)$ (which exists by Lemma \ref{lem:G}) and using that $\Delta  = B + \frac{\delta^2}{4} I_\ell$, we obtain
\begin{align*}
d \dot{\tilde \varphi}_t  = \left(\Delta G(t) - \tfrac{\delta}{2} \dot G(t) \right)^{-1} d \bar M_t + B\left(\tilde \varphi_t- \tilde \xi_t\right)dt + \tfrac{\delta^2}{4} \tilde \varphi_t dt.
\end{align*}
Finally, again taking into account~\eqref{eq:changevar} and defining the square-integrable martingale $M$ by
\begin{equation*}
d M_t = e^{\frac{\delta}{2} t} \left(\Delta G(t) - \tfrac{\delta}{2} \dot G(t) \right)^{-1} d \bar M_t, \quad M_0 = \bar M_0,
\end{equation*}
we obtain that $\varphi$ from  \eqref{eq:varphi} indeed satisfies the FBSDE dynamics (\ref{eq:fbsde1}-\ref{eq:fbsde2}) with terminal condition \eqref{eq:fbsde3:terminal}.
\end{proof}

Let us again briefly comment on the financial interpretation of this result in the context of Lemma~\ref{thm:indopt1}. The basic interpretation is the same as in the infinite-horizon case studied in Theorem~\ref{thm:fbsde:infinite}. However, to account for the terminal condition that the trading speed needs to vanish, the optimal relative trading speed in \eqref{eq:ODE} is no longer constant. Instead, it interpolates between this terminal condition and the stationary long-run value from Theorem~\ref {thm:fbsde:infinite}, that is approached if the time horizon is distant. Analogously, the exponential discounting kernel used to compute the target portfolio in Theorem~\ref{thm:fbsde:infinite} is replaced by a more complex version here; compare \cite{bank.al.17} for a detailed discussion in the one-dimensional case. 

\medskip

To apply Theorems~\ref{thm:fbsde:infinite} and~\ref{thm:fbsde} to characterize the equilibrium in Theorem~\ref{thm:main} it remains to verify that the matrix $B$ appearing there only has real, positive eigenvalues, since this implies that the matrix $\Delta := B + \frac{\delta^2}{4} I_\ell$ has only real eigenvalues greater than $\delta^2/4 \geq 0$.

\begin{lemma}\label{lem:eigen}
Let $\Lambda \in  \mathbb{R}^{d \times d}$ be a diagonal matrix with positive entries $\lambda^1, \ldots, \lambda^d > 0$,  $\Sigma \in \mathbb{R}^{d \times d}$ a symmetric, positive definite matrix and $\gamma^1,\ldots,\gamma^N>0$ with $\gamma^N = \max(\gamma^1,\ldots,\gamma^N)$.  Then, the matrix
$$
B= \begin{pmatrix} \frac{\gamma^N-\gamma^1}{N} \frac{\Lambda^{-1} \Sigma}{2} & \cdots &\frac{\gamma^{N}-\gamma^{N-1}}{N} \frac{\Lambda^{-1} \Sigma}{2}   \\ \vdots & \cdots & \vdots  \\ \frac{\gamma^N-\gamma^1}{N} \frac{\Lambda^{-1} \Sigma}{2}  & \cdots & \frac{\gamma^{N}-\gamma^{N-1}}{N} \frac{\Lambda^{-1} \Sigma}{2}  \end{pmatrix} + \begin{pmatrix} \gamma^1 \frac{\Lambda^{-1} \Sigma}{2} &  &0  \\  & \ddots &   \\ 0 & & \gamma^{N-1} \frac{\Lambda^{-1} \Sigma}{2}  \end{pmatrix} \in \mathbb{R}^{d (N-1)\times d (N-1)}
$$
has only real, positive eigenvalues. 
\end{lemma}

\begin{proof}
First, recall that two matrices that are similar have the same eigenvalues. Since the matrix $\frac{\Lambda^{-1} \Sigma}{2}$ has only positive eigenvalues (because it is the product of two symmetric positive definite matrices, cf.~\cite[Proposition 6.1]{serre.10}), there is an invertible matrix $P \in \mathbb{R}^{d\times d}$ and a diagonal matrix $U \in \mathbb{R}^{d\times d}$ with positive diagonal entries $u^1, \ldots, u^d$ such that $\frac{\Lambda^{-1} \Sigma}{2} = P U P^{-1}$. Now, define the matrix
\begin{equation*}
Q = \begin{pmatrix}
    P &  & 0 \\
 & \ddots & \\
0 &  & P
\end{pmatrix}  \in \mathbb{R}^{d (N-1)\times d (N-1)}.
\end{equation*}
A direct computation shows that $Q$ is invertible with inverse
\begin{equation*}
Q^{-1} = \begin{pmatrix}
P^{-1} &  & 0 \\
& \ddots & \\
0 &  & P^{-1}
\end{pmatrix}  \in \mathbb{R}^{d (N-1)\times d (N-1)}.
\end{equation*}
 Whence $B$ is similar to $\bar B := P^{-1} B P$, and
\begin{equation*}
\bar B = \begin{pmatrix} \frac{\gamma^N-\gamma^1}{N} U& \cdots &\frac{\gamma^{N}-\gamma^{N-1}}{N} U  \\ \vdots & \cdots & \vdots  \\ \frac{\gamma^N-\gamma^1}{N} U & \cdots & \frac{\gamma^{N}-\gamma^{N-1}}{N} U \end{pmatrix} + \begin{pmatrix} \gamma^1 U&  &0  \\  & \ddots &   \\ 0 & & \gamma^{N-1} U \end{pmatrix}.
\end{equation*}
To prove that $\bar B$ (and hence $B$) only has real and positive eigenvalues, we calculate the determinant of $V(x) = xI_{d (N-1)} - \bar B$ for $x \in \mathbb{C} \setminus (0, \infty)$ and show that $\det(V(x)) \neq 0$. So let $x \in \mathbb{C} \setminus (0, \infty)$. Denote by $\mathcal{R}_d$ the commutative subring of all diagonal matrices in $\mathbb{C}^{d\times d}$ and let $\aleph^1, \ldots, \aleph^{N-1}, \gimel^1(x), \ldots \gimel^{N-1}(x) \in \mathcal{R}_d$ be given by
\begin{equation*}
\aleph^n = - \frac{\gamma^N - \gamma^n}{N} U \quad \text{and} \quad \gimel^n(x) = x I_d - \gamma^n U.
\end{equation*}
With this notation, the $\mathbb{R}^{d(N-1) \times d (N-1)}$-valued matrix $V(x)$ can also be understood as an element of $\mathcal{R}_d^{(N-1)\times(N-1)}$ (the $(N-1)\times(N-1)$ matrices with elements from the diagonal matrices in $\mathbb{R}^{d \times d }$) and we have
\begin{equation*}
V(x) =
\begin{pmatrix} 
\aleph^1 +\gimel^1(x) & \aleph^2 & \cdots &\aleph^{n-2} &\aleph^{N-1} \\
\aleph^1 & \aleph^2 +\gimel^2(x)  & \ddots & \vdots &\vdots \\
\vdots & \ddots  & \ddots & \ddots &\vdots \\
\vdots & \vdots & \ddots &\aleph^{N-2}+\gimel^{N-2}(x) &\aleph^{N-1} \\
\aleph^1  & \aleph^2 & \cdots &\aleph^{N-2} &\aleph^{N-1} + \gimel^{N-1}(x)
 \end{pmatrix}.
\end{equation*}
Now use that by \cite[Theorem 1]{silvester.00}, $\det (V (x)) = \det( \mathfrak{det}(V(x)))$, where $ \mathfrak{det}: \mathcal{R}_d^{(N-1)\times(N-1)} \to  \mathcal{R}_d$ is the determinant map on the commutative ring $\mathcal{R}_d$. By subtracting the last row (in $\mathcal{R}_d$) of $V(x)$ from the other rows, the problem boils down to calculating the determinant of 
\begin{equation*}
\bar V(x) := \begin{pmatrix} 
\gimel^1(x) & 0 & \cdots &0 &- \gimel^{N-1}(x) \\
0 & \gimel^2(x)  & \ddots & \vdots &\vdots \\
\vdots & \ddots  & \ddots & 0 &\vdots \\
0 & \cdots &0 &\gimel^{N-2}(x) &- \gimel^{N-1}(x) \\
\aleph^1  & \aleph^2 & \cdots &\aleph^{N-2} &\aleph^{N-1} + \gimel^{N-1}(x)
 \end{pmatrix}.
\end{equation*}
As $x \in \mathbb{C} \setminus (0, \infty)$ and for $n \in \{1, \ldots, N-1\}$, the eigenvalues of $\gamma^n U$ are $\gamma^n u^1, \ldots, \gamma^n u^d \in (0, \infty)$, it follows that $\det(\gimel^n (x)) \neq 0$ and hence $\gimel^n (x)$ is invertible for each $n$. Now, subtracting $\aleph^n (\gimel^n(x))^{-1}$-times the $n$-th row from the last row for $n = 1, \ldots, N-2$, the problem simplifies to calculating the determinant of 
\begin{equation*}
\hat V(x) := \begin{pmatrix} 
\gimel^1(x) & 0 &0 &- \gimel^{N-1}(x) \\
0 & \ddots  & 0 &\vdots \\
0 & \ddots &\gimel^{N-2}(x) &- \gimel^{N-1}(x) \\
0 & \cdots &0 &\gimel^{N-1}(x)\left(I_d + \sum_{n =1}^{N-1} \aleph^n
(\gimel^n(x))^{-1}\right)
 \end{pmatrix}.
\end{equation*}
As a result:
\begin{equation*}
 \mathfrak{det}(\hat V(x)) = \left(\prod_{n = 1}^{N-1} \gimel^n(x) \right) \left(I_d+ \sum_{n =1}^{N-1} \aleph^n
(\gimel^n(x))^{-1}\right),
\end{equation*}
and in turn
\begin{equation*}
\det (V(x)) = \det \Big(  \mathfrak{det}(\hat V(x)) \Big) = \left(\prod_{n = 1}^{N-1} \det(\gimel^n(x))\right) \det\left(I_d + \sum_{n =1}^{N-1} \aleph^n
(\gimel^n(x))^{-1}\right).
\end{equation*}
It therefore remains to show that $ \det\left(I_d + \sum_{n =1}^{N-1} \aleph^n
(\gimel^n(x))^{-1}\right) \neq 0$. As $I_d + \sum_{n =1}^{N-1} \aleph^n
(\gimel^n(x))^{-1}$ is a diagonal matrix, we have
\begin{equation*}
\det\left(I_d + \sum_{n =1}^{N-1} \aleph^n
(\gimel^n(x))^{-1}\right) = \prod_{i = 1}^d \left(1 + \sum_{n =1}^{N-1} \alpha^n u^i \frac{1}{x - \gamma^n u^i}\right),
\end{equation*}
where $\alpha^n = -\frac{\gamma^N -\gamma^n}{N}$, $n \in \{1, \ldots, N-1\}$. As $\gamma^N = \max(\gamma^1, \ldots, \gamma^{N})$, we have $\alpha^n \leq 0$ for each $n \in \{1, \ldots, N-1\}$. It suffices to show that for $i \in \{1, \ldots, d\}$,
\begin{equation}
\label{eqn:eigenvalue:01}
1 + \sum_{n =1}^{N-1} \alpha^n u^i \frac{1}{x - \gamma^n u^i} \neq 0.
\end{equation}
Writing $x = \Re(x) + \mathrm{i} \Im(x)$ and expanding each fraction in \eqref{eqn:eigenvalue:01} by $\Re(x) - \Im(x) \mathrm{i} - \gamma^n u^i$, we obtain
\begin{align}
1 + \sum_{n =1}^{N-1} \alpha^n u^i \frac{1}{x - \gamma^n u^i} &= 1 + \sum_{n =1}^{N-1} \alpha^n u^i \frac{\Re(x) - \gamma^n u^i - \Im(x) \mathrm{i}}{(\Re(x) - \gamma^n u^i)^2 + \Im(x)^2} \notag \\
&= 1 + \Re(x) \sum_{n =1}^{N-1} \alpha^n u^i \frac{1}{(\Re(x) - \gamma^n u^i)^2 + \Im(x)^2} \\
&\quad\quad + \sum_{n =1}^{N-1} \alpha^n u^i \frac{- \gamma^n u^i}{(\Re(x) - \gamma^n u^i)^2 + \Im(x)^2} \notag \\
&\quad \quad - \Im(x) \mathrm{i} \sum_{n =1}^{N-1} \alpha^n u^i \frac{1}{(\Re(x) - \gamma^n u^i)^2 + \Im(x)^2}\notag \\
&= 1 + \Re(x) c^i(x) + d^i(x) - \Im(x) \mathrm{i} c^i(x) \notag \\
&= 1 +  d^i(x) + c^i(x) \overline{x},
\end{align}
where $\overline{x}$ is the complex conjugate of $x$ and
\begin{align*}
c^i(x) &= \sum_{n =1}^{N-1} \alpha^n u^i \frac{1}{(\Re(x) - \gamma^n u^i)^2 + \Im(x)^2}, \\
d^i(x) &= \sum_{n =1}^{N-1} \alpha^n u^i \frac{- \gamma^n u^i}{(\Re(x) - \gamma^n u^i)^2 + \Im(x)^2}.
\end{align*}
As each $\alpha^n$ is nonpositive, $d^i(x)$ is nonnegative and $c^i(x)$ is nonpositive. Combining this with $\overline{x} \in \mathbb{C} \setminus (0, \infty)$, it follows that $1 + d^i(x)  + c^i(x) \overline{x} \neq 0$ so that all eigenvalues of the matrix $B$ are indeed real and positive.
\end{proof}

\section{Primary Matrix Functions}
In this appendix, we collect some facts about matrix functions from the textbook \cite{higham.08} that are used in Appendix~\ref{a:FBSDE}. First, we recall the definition of a (primary) matrix function:

\begin{definition}\normalfont
	\label{def:matrix function} 
Let $A \in \mathbb{C}^{\ell \times \ell}$ be a matrix with distinct eigenvalues $\lambda_1, \dots, \lambda_m$, $m \leq \ell$. Denote by $n_i$  the algebraic multiplicity of $\lambda_i$, $i \in \{1, \ldots, m\}$.  Let $O$ be an open neighbourhood of $\lambda_1, \dots, \lambda_m$ in $\mathbb{C}$ and $f: O \to \mathbb{C}$ a function.\footnote{If $A \in \mathbb{R}^{\ell \times \ell}$ and all eigenvalues of $A$ are real, $O$ can be taken as an open neighbourhood of $\lambda_1, \dots, \lambda_m$ in $\mathbb{R}$, provided that $f$ is also real valued.}
\begin{enumerate}
	\item [{\normalfont(a)}] The function $f$ is said to be \emph{defined on the spectrum of $A$} if it is $n_i-1$-times differentiable at $\lambda_i$, $i \in \{1, \ldots, m\} $.
	\item [{\normalfont(b)}] If $f$ is defined on the spectrum of $A$, then the \emph{primary matrix function} $f(A)$ is defined by
\begin{equation*}
f(A) :=  p(A),
\end{equation*}
where $p:  \mathbb{C} \to \mathbb{C}$ is the unique Hermite interpolating polynomial satisfying $p^{(k)}(\lambda_i) = f^{(k)}(\lambda_i)$ for $k \in \{0, \ldots, n_i-1\}$ and $i \in \{1, \ldots, s\}$.\footnote{If $A$, $\lambda_1, \ldots, \lambda_m$, $O$ and $f$ are all real valued and $f$ defined on the spectrum of $A$, the Hermite interpolating polynomial is also real valued.}
\end{enumerate}
\end{definition}

As a prime example, note that the exponential function is defined on the spectrum of all matrices $A \in \mathbb{C}^{\ell \times \ell}$ and $\exp(A)$ is just the matrix exponential. We recall some elementary properties of (primary) matrix functions:

\begin{lemma}
\label{lem:matrix function:properties}
Let $A \in \mathbb{C}^{\ell \times \ell}$ be a matrix with distinct eigenvalues $\lambda_1, \dots, \lambda_m$, $m \leq \ell$ and $f: \mathbb{C} \to \mathbb{C}$ a function defined on the spectrum of $A$. Then: 
\begin{enumerate}
	\item  [{\normalfont(a)}] If $P \in \mathbb{C}^{\ell \times \ell}$ commutes with $A$, then $f(A)$ and $P$ also commute.
		\item  [{\normalfont(b)}] If $P \in  \mathbb{C}^{\ell \times \ell}$ is invertible, then $P f(A) P^{-1} = f(P A P^{-1})$.
		\item [{\normalfont(c)}] The eigenvalues of $f(A)$ are $f(\lambda_1), \ldots, f(\lambda_m)$.
		\item  [{\normalfont(d)}] $f(A)$ is invertible if and only if $f(\lambda_i) \neq 0$ for all $i \in \{1, \ldots, m\}$.
\end{enumerate}
\end{lemma}

\begin{proof}
Assertions (a), (b) and (c) are parts of \cite[Theorem 1.13]{higham.08}. Finally, (d) follows from (c) and the fact that $f(A)$ is invertible if and only if zero is not an eigenvalue. 
\end{proof}

Finally, we recall a result on the principal square root~\cite[Theorem 1.29]{higham.08}:
\begin{lemma}
	\label{lem:principal square root}
Let $A \in \mathbb{R}^{\ell \times \ell}$ be a matrix whose eigenvalues are all real and positive. Then there exists a unique matrix $P \in \mathbb{R}^{\ell \times \ell}$ with positive eigenvalues such that $P^2 = A$. It is given by the primary matrix function $P = \sqrt{A}$ in the sense of Definition \ref{def:matrix function}. 
\end{lemma}

\bibliographystyle{abbrv}
\bibliography{mms}

\end{document}